\documentclass[10pt, twocolumn]{IEEEtran}

\usepackage{amsthm}
\usepackage{amsmath}
\usepackage{amssymb}
\usepackage{amsfonts}
\usepackage{epsfig}
\usepackage{subfigure}
\usepackage{calc}
\usepackage{color}
\usepackage[all]{xy}
\usepackage{bm}
\usepackage{enumerate}
\usepackage{pdfsync}
\usepackage{framed}
\usepackage{hyperref}

\newtheorem{defn}{Definition}
\newtheorem{fact}{Fact}
\newtheorem{thm}{Theorem}[section]
\newtheorem{cor}[thm]{Corollary}
\newtheorem{prop}{Proposition}

\newtheorem{lem}[thm]{Lemma}
\newtheorem{conj}[thm]{Conjecture}
\newtheorem{constr}[thm]{Construction}
\newtheorem{note}{Remark}
\newtheorem{example}{Example}
\newcommand{\bit}{\begin{itemize}}
\newcommand{\eit}{\end{itemize}}
\newcommand{\bcor}{\begin{cor}}
\newcommand{\ecor}{\end{cor}}
\newcommand{\beq}{\begin{equation}}
\newcommand{\eeq}{\end{equation}}
\newcommand{\beqn}{\begin{equation*}}
\newcommand{\eeqn}{\end{equation*}}
\newcommand{\bea}{\begin{eqnarray}}
\newcommand{\eea}{\end{eqnarray}}
\newcommand{\bean}{\begin{eqnarray*}}
\newcommand{\eean}{\end{eqnarray*}}
\newcommand{\ben}{\begin{enumerate}}
\newcommand{\een}{\end{enumerate}}
\newcommand{\bdefn}{\begin{defn}}
\newcommand{\edefn}{\end{defn}}
\newcommand{\bnote}{\begin{note}}
\newcommand{\enote}{\end{note}}
\newcommand{\bprop}{\begin{prop}}
\newcommand{\eprop}{\end{prop}}
\newcommand{\blem}{\begin{lem}}
\newcommand{\elem}{\end{lem}}
\newcommand{\bthm}{\begin{thm}}
\newcommand{\ethm}{\end{thm}}
\newcommand{\bconj}{\begin{conj}}
\newcommand{\econj}{\end{conj}}
\newcommand{\bconstr}{\begin{constr}}
\newcommand{\econstr}{\end{constr}}
\newcommand{\bpf}{\begin{proof}}
\newcommand{\epf}{\end{proof}}

\begin{document}

\title{Communication Cost for Updating Linear Functions when Message Updates are Sparse: Connections to Maximally Recoverable Codes}
\author{N. Prakash and Muriel M{\'{e}}dard
\thanks{N. Prakash and Muriel M{\'{e}}dard are with the Research Laboratory of Electronics, Massachusetts Institute of Technology, USA (email: \{prakashn,  medard\}@mit.edu).}
\thanks{The results in this article were presented in part as an invited talk at 53rd Annual Allerton Conference on Communication, Control, and Computing, Sept 29-Oct 2, 2015, Allerton Park and Retreat Center, Monticello, IL, USA. The work is in part supported by the AFOSR award  FA9550-14-1-0403.}
}
\maketitle

\begin{abstract}
We consider a communication problem in which an update of the source message needs to be conveyed to one or more distant receivers that are interested in maintaining specific linear functions of the source message. The setting is one in which the updates are sparse in nature, and where neither the source nor the receiver(s) is aware of the exact {\em difference vector}, but only know the amount of sparsity that is present in the difference-vector. Under this setting, we are interested in devising linear encoding and decoding schemes that  minimize the communication cost involved. We show that the optimal solution to this problem is closely related to the notion of maximally recoverable codes (MRCs), which were originally introduced in the context of coding for storage systems. In the context of storage, MRCs guarantee optimal erasure protection when the system is partially constrained to have local parity relations among the storage nodes. In our problem, we show that optimal solutions exist if and only if MRCs of certain kind (identified by the desired linear functions) exist.  We consider point-to-point and broadcast versions of the problem, and identify connections to MRCs under both these settings. For the point-to-point setting, we show that our linear-encoder based achievable scheme is optimal even when non-linear encoding is permitted. The theory is illustrated in the context of updating erasure coded storage nodes. We present examples based on modern storage codes such as the minimum bandwidth regenerating codes.
\end{abstract}

\section{Introduction}\label{sec:sys_model}

We consider a communication problem in which an update of the source message needs to be communicated to one or more distant receivers that are interested in maintaining specific functions of the actual source message. The setting is one in which the updates are sparse in nature, and where neither the source nor receivers are aware of the \textit{difference-vector}. Under this setting, we are interested in devising encoding and decoding schemes that allow the receiver to update itself to the desired function of the new source message, such that the communication cost is minimized.  

\subsection{Pont-to-Point Setting}

The system for the point-to-point case is shown in Fig. \ref{fig:sys_model}. Here, the $n$-length column-vector $X^n \in \mathbb{F}_q^n$ denotes the initial source message, where $\mathbb{F}_q$ denotes the finite field of $q$ elements. The receiver maintains the linear function $AX^n$, where $A$ is an $m \times n$ matrix, $m \leq n$, over $\mathbb{F}_q$. Let the updated source message be given by $X^n+E^n$, where $E^n$ denotes the difference-vector, and is such that $\text{Hamming wt.}(E^n) \leq \epsilon$, $0 \leq \epsilon \leq n$. We say that the vector $E^n$ is $\epsilon$-sparse. We consider linear encoding at the source using the $\ell \times n$ matrix $H$.  We assume that the source is aware of the function $A$ and the parameter $\epsilon$, but does not know the vector $E^n$. The goal is to encode $X^n + E^n$ at the source such that the receiver can update itself to $A(X^n  +E^n)$, given the source encoding and $AX^n$. Assuming that the parameters $n, q, \epsilon$ and $A$ are fixed, we define the communication cost as the parameter $\ell$, which is the number of symbols generated by the linear encoder $H$. The goal is to design the encoder and decoder so to minimize $\ell$. We assume a zero-probability-of-error, worst-case-scenario model; in orther words, we are interested in schemes which allow error-free decoding for every $X^n \in \mathbb{F}_q^n$ and $\epsilon$-sparse $E^n \in \mathbb{F}_q^n$.

\begin{figure*}[h]
	\centering
	\includegraphics[width=5in]{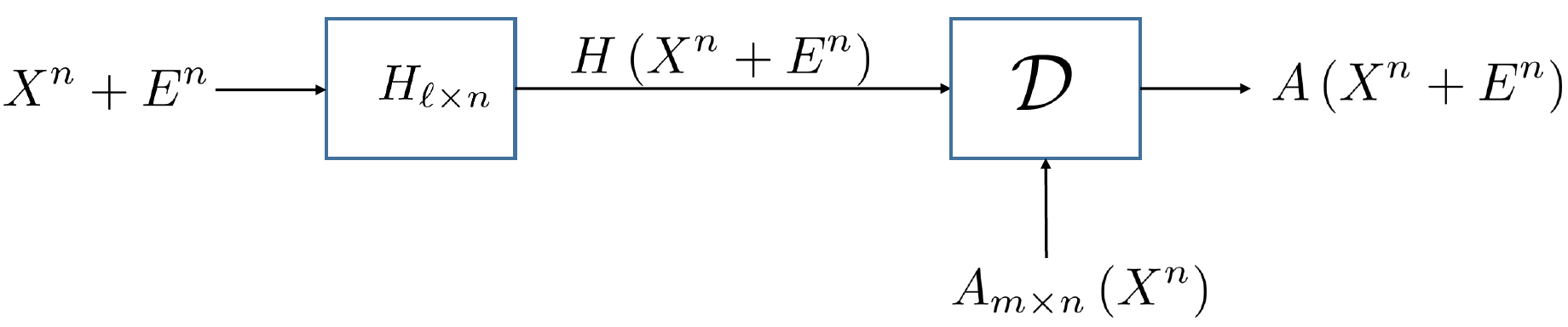}
	\caption{The system model for function update in the point-to-point setting.}
	\label{fig:sys_model}
\end{figure*}

The problem is in part motivated by the setting of distributed storage systems (DSSs)  that use linear erasure codes for data storage, and where the underlying uncoded file is subject to frequent but tiny updates. Modern day DSSs support multiple users to access and update any single file. Users often do not communicate with each other. In this scenario, it is the case that a user applies his updates to a local copy of the file that is possibly different from the one present in the system (since the file might have been updated by other users). Consequently, the user is unaware of the precise difference vector between the file which the user wants to upload and the file present in the DSS. One possible solution is for the user to first learn the most recent file version that is present in the system, and then apply the update on that file. Such a solution has two limitations. Firstly, the solution is blocking in the sense that any time, at most one user can update the file, and any other user must wait until either previous users complete or time out. This is not a preferred approach when a large number of users need to be supported. Secondly, for the user to learn the most recent file, the user often needs to decode an erasure coded file (or a part of it) that can be substantially larger than the size of the update, and this can be computationally expensive. This is because for cost-effective erasure coding, DSSs demand a certain minimum size for the data file (or part of it) that is getting erasure coded. For instance, Microsoft Giza~\cite{giza}, which is a geo-distributed file storage system, considers erasure coding only those data stripes that are larger than 4 MB. In this case, the user at least needs to decode a 4 MB file chunk even if the update is of the order of a few kBs. 

In this work, we propose a solution that does not have either of the above two limitations. Our solution works based on the assumption that the user, who wants to update the coded file, has knowledge of a bound on size of the update. We consider updates as substitutions in the file, which is a reasonable assumption in scenarios when the file-size is much larger than the quantum of updates. With respect to Fig. \ref{fig:sys_model}, $X^n$ represents the version of the file that is stored in the DSS. The user  has access to a locally obtained version $X^n + E_1^n$ that is not stored in the DSS, and updates it to $X^n + E_1^n + E_2^n = X^n + E^n$, where $E^n = E_1^n + E_2^n$ represents the overall update that the file has undergone. Further, our interest is in DSSs like peer-to-peer DSS where the user  establishes multiple connections with various nodes in order to update their coded data.  Figure \ref{fig:sys_model} illustrates a scenario in which the user updates the coded data corresponding to one of the storage nodes that uses matrix $A$ as its erasure coding coefficients. We illustrate the usage of the matrix $A$ via the following example.

\vspace{0.1in}

\begin{example} \label{ex:moti}
Consider a DSS that uses an $[N, K]$ linear code for storing data across $N$ storage nodes. The data file $X^n$ is striped before encoding, the various stripes are individually encoded and stacked against each other to get the overall coded file. Let ${\bf a} = [a_1 \ a_2 \ \cdots a_K ], a_i \in \mathbb{F}_q$ denote the coding coefficients corresponding to one of the storage nodes. Assuming that the first $K$ symbols of the vector $X^n$ correspond to the first stripe, symbols $K+1$ to $2K$ correspond to the second stripe and so on, the overall coding matrix $A$ (taking into account all stripes) is given by
\begin{eqnarray}
A & = & \left[ \begin{array}{cccc}
{\bf a} & & & \\ & {\bf a} & &  \\ & & \ddots & \\ & & & {\bf a}
\end{array} \right].
\end{eqnarray} 
Here we assume that $K|n$, and thus there are $m = \frac{n}{K}$ stripes in the file. The parameter $m$ also represents the number of coded symbols stored by any one of the $N$ nodes. The matrix $A$ is of size $m \times mK$. The scenario of interest is one in which the number of coded symbols $m$ is much larger than the sparsity parameter $\epsilon$, which is the number of symbols that get updated in the uncoded file. The goal is to design the encoder $H$ and the decoder $\mathcal{D}$  that minimizes the communication cost. 
\end{example}

\subsection{Broadcast Setting}

The system model naturally extends to broadcast settings consisting of a single source and multiple destinations that are interested in possibly different functions of the source message. In our work, we study the broadcast setting for the case of two destinations (see  Fig. \ref{fig:sys_model_diff_function}), where the receivers initially hold functions $AX^n$ and $BX^n$, and get updated to $A(X^n + E^n)$ and $B(X^n + E^n)$, respectively. The broadcast setting is motivated by settings in which two of the storage nodes in a peer-to-peer DSS are nearly co-located, as far as a (distant) user is concerned. In this case, the question of interest is to identify whether it is possible to send one encoded packet $H(X^n + E^n)$ to simultaneously update the coded data of both the storage nodes. 
We will next present two examples of the broadcast setting. As we will see later in this document, the first example is an instance in which broadcasting does not help to reduce communication cost, i.e., it's optimal to individually update the two destinations. The second example is an instance in which it is beneficial to broadcast than individually update the two destinations.

\begin{figure*}[h]
	\centering
	\includegraphics[width=5in]{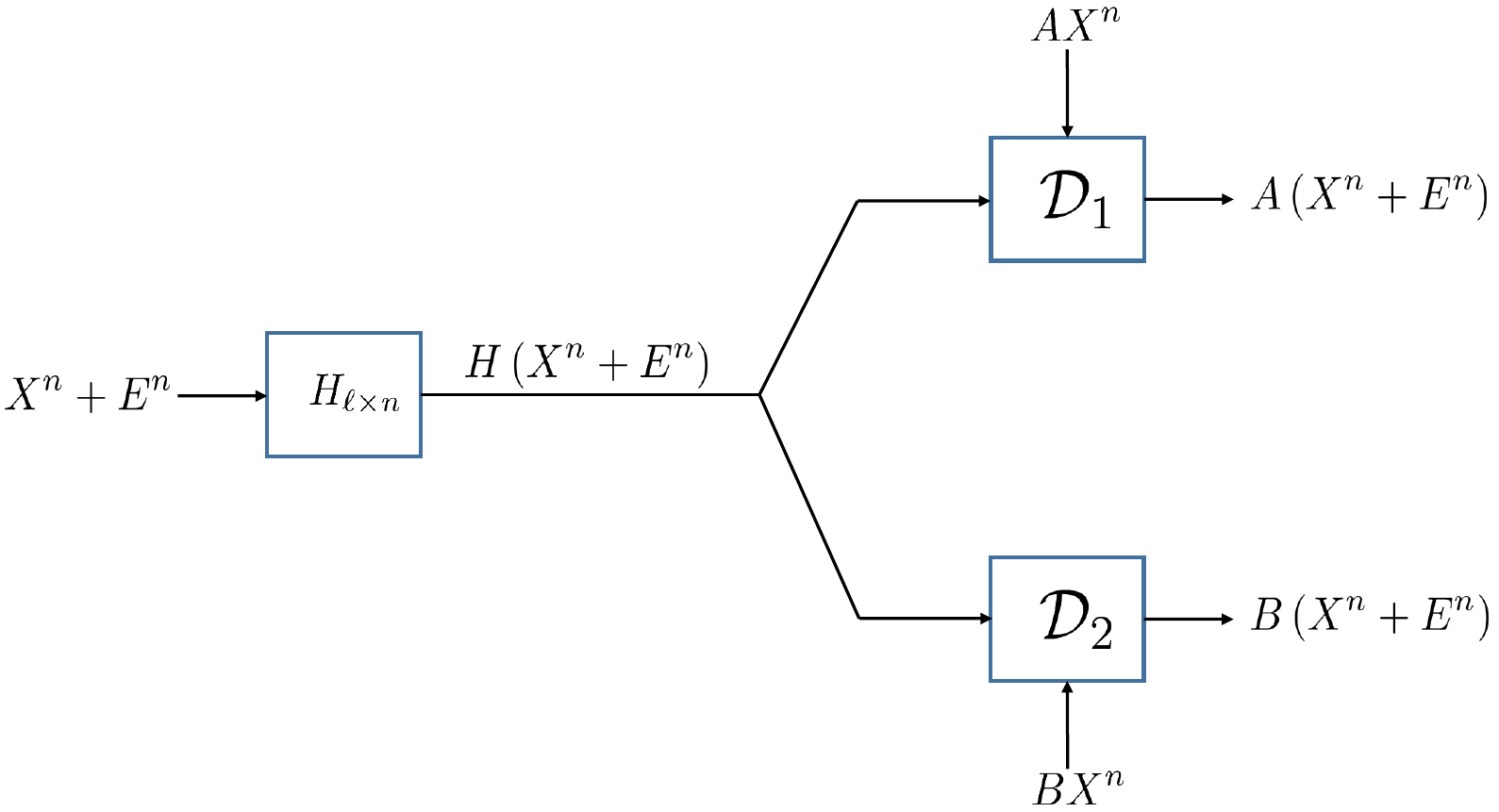}
	\caption{The system model for function updates in the broadcast setting involving two receivers. }
	\label{fig:sys_model_diff_function}
\end{figure*}

\vspace{0.1in}

\begin{example} \label{ex:bcast_nouse}
As in Example \ref{ex:moti}, we consider striping and encoding of a data file by an $[N, K]$ linear code over $\mathbb{F}_q$ with $K \geq 2$. We apply the broadcast model to simultaneously update the contents of $2$ out of the $N$ storage nodes. Let us assume that the $K$-length coding vectors associated with two of the storage nodes are given by ${\bf a} = [a_1 \ a_2 \ \ldots a_K], {\bf b} = [b_1 \ b_2 \ \ldots b_K], a_i, b_i \in \mathbb{F}_q, {\bf a} \neq {\bf b}$. The overall coding matrices $A$ and $B$ corresponding to the $m$ stripes are then given by 
	\begin{eqnarray} \label{eq:ABspecial}
	A \  = \ \left[ \begin{array}{cccc}
	{\bf a} & & & \\ & {\bf a} & &  \\ & & \ddots & \\ & & & {\bf a}
	\end{array} \right] \hspace{-0.1in} & \text{and} & \hspace{-0.1in}
	B \  = \ \left[ \begin{array}{cccc}
	{\bf b} & & & \\ & {\bf b} & &  \\ & & \ddots & \\ & & & {\bf b}
	\end{array} \right],
	\end{eqnarray} 
	where both $A$ and $B$ have size $m \times mK$. We will see later that for this example,  in order to achieve the optimal communication cost, the source must transmit as though it is individually updating the two storage nodes, i.e., there is no added benefit due to broadcasting.
\end{example}

\vspace{0.1in}

\begin{example} \label{ex:bcast_use}
	In this example, we consider striping and encoding of a data file by a \textit{minimum bandwidth regenerating} (MBR) code~\cite{dimakis} that is obtained via the product-matrix  construction described in \cite{prod_matrix}. Regenerating codes~\cite{dimakis} are codes specifically designed for data storage, and they allow the system designer to trade-off storage overhead against \textit{repair bandwidth} for a given level of fault tolerance. MBR codes have the best possible repair bandwidth (at the expense of storage overhead) for a given level of fault tolerance. We first give a quick introduction to a specific instance of an MBR code in \cite{prod_matrix}, and use that in our broadcast setting.
	
	\vspace{0.05in}
	
	{\em The MBR code:} Consider an $[N = 5, K = 3, D = 4] (\alpha = 4, \beta = 1)$ MBR code over $\mathbb{F}_q$ that encodes $9$ symbols into $20$ symbols and stores across $N = 5$ nodes such that each node holds $\alpha = 4$ symbols. The code has the property that the contents of any $K = 3$ nodes are sufficient to reconstruct the $9$ uncoded symbols. Also, the contents of any one of the $N = 5$ nodes can be recovered by connecting to any set of $D = 4$ other nodes\footnote{In our example, one connects to all the remaining $N -1 = 4$ nodes, since $D = 4$.} and downloading $\beta = 1$ symbols from each of them. The latter property is known as the node-repair property of the regenerating code. Let $[m_1, \ldots, m_9]$ denote the vector of $9$ message symbols, where $m_i \in \mathbb{F}_q$. The encoding is described by a product-matrix construction as follows:
	\begin{eqnarray}
		C_{N \times \alpha} & = & \Psi_{N \times D} M_{D \times \alpha} \\
		& = & \left[\begin{array}{c} \underline{\psi}_1 \\  \underline{\psi}_2 \\ \underline{\psi}_3 \\ \underline{\psi}_4 \\ \underline{\psi}_5 \end{array} \right] \left[ \begin{array}{cccc} m_1 & m_2 & m_3 & m_7 \\ m_2 & m_4 & m_5 & m_8 \\ m_3 & m_5 & m_6 & m_9 \\ m_7 & m_8 & m_9 & 0 \end{array} \right]. \label{eq:MBR_prod_mtx}
	\end{eqnarray}
	Here, $M$ denotes the $D \times \alpha$ message matrix whose entries are populated from among the $9$ message symbols, in a specific manner, as shown in \eqref{eq:MBR_prod_mtx}. Note that the matrix $M$ is symmetric. The $N \times D$ encoding matrix $\Psi$ can be chosen as a Vandermonde matrix under the product-matrix framework. The vector $\underline{\psi}_i, 1 \leq i \leq 5$ denotes the $i^{\text{th}}$ row of $\Psi$. In this example, we assume that $\Psi$ is given as follows:
	\begin{eqnarray}
		\Psi & = & \left[ \begin{array}{cccc} 1 & \gamma & \gamma^2 & \gamma^3 \\  1 & \gamma^2 & \gamma^4 & \gamma^6 \\ 1 & \gamma^3 & \gamma^6 & \gamma^9 \\ 1 & \gamma^4 & \gamma^8 & \gamma^{12} \\ 1 & \gamma^5 & \gamma^{10} & \gamma^{15} \end{array} \right],
	\end{eqnarray}
	where we pick $\gamma$ as a primitive element in $\mathbb{F}_q$. The $N \times \alpha$ matrix $C$ is the codeword matrix, with the $i^{\text{th}}$ row representing the contents that get stored in the $i^{\text{th}}$ node, $1 \leq i \leq 5$. 
	
	\vspace{0.05in}
		
	{\em Striping using the MBR code, and the associated coding matrices $A$ and $B$:}
	As before, we let $X^n$ to denote the uncoded data file, which is divided into $m$ stripes each of length $9$ symbols. Also, recall that in our model of striping, the first stripe (for the current example) consists of the first $9$ symbols of $X^n$, the second stripe consists of symbols ${X_{10}, \ldots, X_{18}}$, and so on. The length $n$ is given by $n=9m$. Consider the case where we use the broadcast model to update the contents of the first two nodes. The overall coding matrices $A$ and $B$, both of size $4m \times 9m$, corresponding to the first and the second node are respectively given by 
	
	\begin{eqnarray}
		A  =  \left[ \begin{array}{cccc} A_1 & & & \\ & A_1 & & \\ & & \ddots & \\ & & & A_1 \end{array} \right]\hspace{-0.1in}  & , \hspace{-0.1in} & B  =  \left[ \begin{array}{cccc} B_1 & & & \\ & B_1 & & \\ & & \ddots & \\ & & & B_1 \end{array} \right], \nonumber
	\end{eqnarray}
	where 
	
		\begin{equation}
			A_1 =  \left[ \begin{array}{ccccccccc} 1 & \gamma & \gamma^2 & & & & \gamma^3 & &  \\ 
				& 1 & & \gamma & \gamma^2 & & &\gamma^3 &   \\ 
				& & 1& & \gamma & \gamma^2 & & & \gamma^3 \\ 
				& & & & & & 1 & \gamma & \gamma^2 \end{array}  \right],
		\end{equation}
		and 
		\begin{equation}
			B_1 =  \left[ \begin{array}{ccccccccc} 1 & \gamma^2 & \gamma^4 & & & & \gamma^6 & &  \\ 
				& 1 & & \gamma^2 & \gamma^4 & & &\gamma^6 &   \\ 
				& & 1& & \gamma^2 & \gamma^4 & & & \gamma^6 \\ 
				& & & & & & 1 & \gamma^2 & \gamma^4 \end{array}  \right].
		\end{equation}
As we shall later, such broadcasting, rather than individually transmitting to the two destinations, can reduce the total communication cost for this example by nearly $12\%$. 
\end{example}
	
\subsection{Connection to Maximally Recoverable Codes}

In this paper, we identify necessary and sufficient conditions for solving the function update problems in Fig. \ref{fig:sys_model}, Fig. \ref{fig:sys_model_diff_function}, for general matrices $A$ and $B$. We show that the existence of optimal solutions under both these settings is closely related to the notion of {\em maximally recoverable codes}, which were originally studied in the context of coding for storage systems~\cite{mrc_first_paper}, \cite{blaum_mds}, \cite{Parikshit_mrc}.  In the context of storage, MRCs form a subclass of a broader class of codes known as locally repairable codes (LRC)~\cite{GopHuaSimYek}, \cite{PapDim}. An $[n, k]$ linear code $\mathcal{C}$ is called an LRC with locality $r$ if each of the $n$ code symbols is recoverable as a linear combination of at most $r$ other symbols of the code. Qualitatively, an LRC is said to be maximally recoverable (a formal definition appears in Section \ref{sec:mrc}) if it offers optimal {\em beyond-minimum-distance} correction capability. When restricted to the setting of LRCs, MRCs are known as  partial-MDS codes~\cite{blaum_mds} or maximally recoverable codes with locality~\cite{Parikshit_mrc}. MRCs with locality are used in practical DSSs like Windows Azure~~\cite{HuaSimXu_etal_azure}.

The usage of the concept of MRCs in our work is more general, and not necessarily restricted to the context of LRCs. Below, we define the notion of {\em maximally recoverable subcode} of a given code, say $\mathcal{C}_0$.

\vspace{0.1in}

\begin{defn}[Maximally Recoverable Subcode $\mathcal{C}$ of $\mathcal{C}_0$] \label{defn:mrsc}
Let $\mathcal{C}_0$ denote an $[n, t]$ linear code over $\mathbb{F}_q$ having a generator matrix $G_0$. Also, consider an $[n, k]$ subcode $\mathcal{C}$ of $\mathcal{C}_0$ for some $k \leq t$, and let $G$ denote a generator matrix of $\mathcal{C}$. The code $\mathcal{C}$ will be referred to as a maximally recoverable subcode (MRSC) of $\mathcal{C}_0$ if for any set $S \subseteq [n], |S| = k$ such that $\text{rank}\left(G_0|_S\right) = k$, we have $\text{rank}\left(G|_S\right) = k$. Here $G_0|_S$ and $G|_S$ respectively denote the submatrices obtained by restricting $G_0$ and $G$  to the columns specified by $S$.
\end{defn}

\vspace{0.1in}

\begin{example}
Consider an $[n = 9, t = 3]$ binary code $\mathcal{C}_0$ whose generator  matrix $G_0$ is given by 	
\begin{eqnarray}
G_0 & = & \left[ \begin{array}{ccccccccc}  1 & 1 & 1 & & & & & &  \\ & & & 1 & 1 & 1 & & & \\ & & & & &  & 1 & 1 & 1 \end{array} \right].
\end{eqnarray}
Next, consider the $[n = 9, k = 2]$ subcode $\mathcal{C}$ of $\mathcal{C}_0$ having a generator matrix $G$ given by 
\begin{eqnarray}
G & = & \left[ \begin{array}{ccccccccc}  1 & 1 & 1 & & & & 1 & 1 & 1 \\ & & & 1 & 1 & 1 & 1 & 1 & 1 \end{array} \right].
\end{eqnarray}
It is straightforward to verify that $\mathcal{C}$ is an MRSC of $\mathcal{C}_0$.
\end{example}	

\subsection{Summary of Results}
Following is a summary of the results obtained in this paper.
\begin{enumerate}
\item  For the point-to-point setting in Fig. \ref{fig:sys_model}, the communication cost is lower bounded by\footnote{The bound in  \eqref{eq:optimal_cost_p2p} was  already  proved in \cite{Shah_oblivious} for an average case setting.}
\begin{eqnarray} \label{eq:optimal_cost_p2p}
\ell & \geq & \min(2\epsilon, \text{rank}(A)).
\end{eqnarray}
If $C_A$ and $C_H$ denote $n$-length linear codes respectively generated by the rows of the matrices $A$ and $H$, then under optimality (i.e., achieving equality in \eqref{eq:optimal_cost_p2p}), we show that the code  $C_H$ must necessarily be a MRSC of $C_A$. An achievable scheme based on MRSCs is presented to establish the optimality of the bound in \eqref{eq:optimal_cost_p2p}. For a general matrix $A$,  our achievability is guaranteed only when the field size $q$ used in the model (see Fig. \ref{fig:sys_model}) is sufficiently large.
\item For the point-to-point case, we obtain a lower bound on the communication cost permitting non-linear encoding. When $\text{rank}(A) \geq 2 \epsilon$, we show that \eqref{eq:optimal_cost_p2p} is a valid lower bound even under non-linear encoding. This further proves that when the field size is sufficiently large, linear encoding suffices to achieve optimal communication cost for the problem. Our proof technique here is based on the analysis of the chromatic number of the characteristic graph associated with the function computation problem~\cite{korner1998zero}, \cite{doshi2010functional}. 
\item  An explicit low field size encoding matrix $H$ is provided for the setting considered in Example \ref{ex:moti}. Recall that in Example \ref{ex:moti}, we use an arbitrary $[N, K]$ linear code to stripe and store the data. The MRSC $C_H$ in this case  corresponds to MRCs with locality, where the local codes are {\em scaled repetition codes}. 
\item We identify necessary and sufficient conditions for solving the broadcast setting given in Fig. \ref{fig:sys_model_diff_function}. Let $C_A$ and $C_B$ respectively denote the linear block codes generated by the rows of the matrices $A$ and $B$. For the special case when the codes $C_A$ and $C_B$ intersect trivially, there is no benefit from broadcasting, i.e., the  encoder must transmit as though it is transmitting individually to the two receivers, and thus the optimal communication is given by 
\begin{eqnarray} \label{eq:optimal_cost_broadcast}
\ell & \geq & \min(2\epsilon, \text{rank}(A)) + \min(2\epsilon, \text{rank}(B)).
\end{eqnarray}
For the general case when $C_A$ and $C_B$ have a non-trivial intersection, the optimal communication cost can be less than what is given by \eqref{eq:optimal_cost_broadcast}. The expression for the optimal communication cost appears in Theorem \ref{thm:cost_broadcast}. Like in the point-to-point case, optimality is guaranteed only when the field size $q$ is sufficiently large.
\item Our achievability scheme for the general case of the broadcast setting involves finding answer to the following sub-problem : Given an $[n, t]$ code $\mathcal{C}_0$ and an $[n, s]$ subcode $\widehat{\mathcal{C}}$ of $\mathcal{C}_0$, can we find an $[n, k]$ MRSC $\mathcal{C}$ of $\mathcal{C}_0$ such that $\widehat{\mathcal{C}}$ is a subcode of $\mathcal{C}$. We refer to this as the problem of finding {\em sandwiched MRSCs}. Here the parameters $s, k, t$ satisfy the relation $ s \leq k \leq t$. We present two different techniques that show the existence of sandwiched MRSCs (under certain trivial necessary conditions on the code $\mathcal{C}_0$). The first technique is a parity-check matrix based approach, and relies on the existence of non-zero evaluations of a multivariate polynomial over a large enough finite field. The second technique is a generator matrix based approach, uses the concept of {\em linearized polynomials}, and yields an explicit construction. 
\end{enumerate}

Finally, we also discuss the possibility of extending the results to settings involving more than $2$ destinations. We identify lower bounds for the communication cost under special cases for the general problem. The achievability proof appears hard. We comment on the technical challenges that  need to be addressed to solve the achievability for the general case. 
	
\subsection{Related Work}

{\em Coding for File Synchronization :}  In \cite{Shah_oblivious}, the authors consider the problem of oblivious file synchronization, under the substitution model, in a distributed storage setting employing linear codes.   The model is one in which one of the coded storage nodes gets updated with the help of other coded storage nodes, who have already undergone updates. The authors considers the  problem of jointly designing the storage code, and also the update scheme. The updates, like in our setting, is carried out under the assumption that only the sparsity of updates is known, and not the actual updates themselves. Optimal solutions are presented when the storage code is assumed to an MDS code. A recent version of this work \footnote{This version was communicated to us by the authors of \cite{Shah_oblivious_recent}, after we uploaded a version of this paper in arxiv.} \cite{Shah_oblivious_recent} extends the work to the setting of regenerating codes, with the restriction that only one symbol ($\epsilon = 1$) is obliviously updated. The key difference between these works and ours is that we consider designing optimal update schemes for an arbitrary linear function, while the above works assume a specific structure of the storage code (e.g., MDS in \cite{Shah_oblivious}). Please also see Remark \ref{note:Shah} for a comparison of the converse statements appearing in the two papers. 

A second work which is related to ours appears in \cite{goparaju_edits}, where the authors consider the problem of designing protocols for simultaneously optimizing storage performance as well as communication cost for updates. Contrary to our setting, the updates are modeled as insertions/deletions in the file.  A similarity with our work (like the work in \cite{Shah_oblivious}) is that the destination node holds a coded form of the data (which can be considered as linear function of the uncoded data), and is interested in updating the coded data. The goal is to communicate and update the coded file in such a way that reconstruction/repair properties in the storage system are preserved. Their protocol permits modifications to the structure of the storage code itself for optimizing the communication cost. Note that in our model we optimize the communication cost under the assumption that the destination is interested in the {\em same} function $A$ of the updated data as well.

The problem of synchronizing files  under the insertion/deletion model has been previously studied, using information-theoretic techniques, in \cite{MaRamTse}, \cite{cadambe_edits}. Recall that our model uses substitutions, rather than insertions/deletions, for specifying updates. In \cite{MaRamTse}, a point-to-point setting is considered in which the decoder has access to a partially-deleted version of the input sequence, as side information. The authors calculate the minimum rate at which the source needs to compress the (undeleted) input sequence so as to permit lossless recovery (i.e., vanishing probability of error for large block-lengths) at the destination.  The  model of \cite{cadambe_edits} allows both insertions and deletions in the original file to get the updated file. Further, both the source and the destination have access to the original file. Under this setting, the authors study the minimum rate at which the updated file must be encoded at the source, so that lossless recovery is possible at the destination.  The problem of synchronizing edited sequences is also considered  in \cite{OrlVis}, where the authors allow insertions, deletions or substitutions, and assume a zero-error recovery model. An important difference between our work and all the works mentioned above is that while we are interested in a specific linear function of the input sequence, all of the above works assume recovery of the entire source sequence at the destination.

\vspace{0.1in}

{\em Maximally Recoverable Codes:} The notion of maximal recoverability in linear codes was originally introduced in \cite{mrc_first_paper}. Low field size constructions of partial MDS codes (another name for MRCs with locality) for specific parameter sets codes appear in \cite{blaum_mds}, \cite{pmds_twosectors_blaum}, \cite{shum_pmds}. In the terminology of partial MDS codes, these three works respectively provide low field size constructions for the settings up to one, two and three global parities, and any number of local parities. Explicit constructions based on linearized polynomials for the case of one local parity and any number of global parities appear in \cite{Parikshit_mrc}. Identifying low field size constructions of partial MDS codes for general parameter sets  remain an open problem. Various approximations of partial MDS codes like Sector Disk codes~\cite{blaum_mds}, \cite{pmds_twosectors_blaum}, \cite{shum_pmds}, STAIR codes~\cite{stair}, and partial-maximally-recoverable-codes~\cite{balaji_pvk} have been considered in literature, and these permit low field size constructions for larger parameter sets. Other known results on maximally recoverable codes include expressions for the weight enumerators of MRCs with locality~\cite{LalLok}, and the generalized Hamming weights (GHWs) of the MRSC $\mathcal{C}$ in terms of the GHWs of the code $\mathcal{C}_0$~\cite{PraLalKum}.

\vspace{0.1in}

{\em Zero-Error Function Computation : } Finally, we note that the problem considered in this paper can be considered as one of zero-error function computation. Even though the existing literature does not directly address the problem that we study here, we review some of the relevant works on zero-error function computation, in point-to-point as well as network settings. In \cite{KowKum2}, the authors examine the problem of computing a general function  with zero-error, under worst-case as well as average-case models in the context of sensor networks. One of the problems considered there involves a point-to-point setting, where a source having access to $x$ needs to encode and transmit to a destination which has access to side information $y$. The destination is interested in computing the function $f(x, y)$ with zero-error in the worst-case model. Though the point-to-point model in Fig. \ref{fig:sys_model} can be considered as a special case of the setting in \cite{KowKum2}, our assumptions regarding the specific nature of the decoder side-information, and  the use of linear encoding enable us to show deeper results, especially the connection to maximally recoverable codes. The problem of zero-error computation of symmetric functions in a network setting is studied in \cite{GirKum}, \cite{KowKum}. In \cite{GirKum}, the authors characterize the rate at which symmetric functions (e.g. mean, mode, max etc.) can be computed at a sink node, while in \cite{KowKum}, the authors provide optimal communication strategies for computing symmetric Boolean functions. The works of \cite{AppFraKarZeg}, \cite{RaiDey}, \cite{RamLan} study zero-error function computation in networks under the framework of network coding. In \cite{AppFraKarZeg}, the authors  study the significance of the min-cut of the network when a destination node is interested in computing a general function of a set of independent source nodes.  The works of \cite{RaiDey}, \cite{RamLan} consider sum networks, where a set of destination nodes in the network is interested in computing the sum of the observations corresponding to a set of source nodes. 

\vspace{0.1in}

The organization of the rest of the paper is as follows.  Sections \ref{sec:mrc} contains definitions and facts relating to maximally recoverable codes. Necessary and sufficient conditions for achieving optimal communication cost in the point-to-point case appears in Section \ref{sec:p2p}. In Section \ref{sec:p2p-example}, we present a low field size construction for the setting considered in Example \ref{ex:moti}. The broadcast setting and the associated problem of finding {\em sandwiched} MRSCs are discussed in Section \ref{sec:broadcast}. Finally, our conclusions appear in Section \ref{sec:conc}.
	
\vspace{0.2in}

{\em Notation} : Given a matrix $A \in \mathbb{F}_q^{m \times n}, m \leq n$ having rank $m$, we write   $\mathcal{C}_A$ to denote the $[n, m]$ linear block code over $\mathbb{F}_q$ generated by the rows of $A$. The  rank of matrix $A$ will be denoted as $\rho(A)$. We  write $\mathcal{C}_A^{\perp}$ to denote the dual code of $\mathcal{C}_A$. For any set $S \subseteq [n] = \{1, 2, \ldots, n\}$ and for any $n$-length linear code $\mathcal{C}$, we use $\mathcal{C}|_S$ to denote the restriction of the code $\mathcal{C}$ to the set of coordinates indexed by the set $S$. Also, we  write $\mathcal{C}^S$ to denote the subcode of $\mathcal{C}$ obtained by {\em shortening} $\mathcal{C}$ to the set of coordinates indexed by the set $S$. In other words, $\mathcal{C}^S$ denotes the set of all those codewords in $\mathcal{C}$ whose {\em support} is confined to $\mathcal{S}$. For any codeword ${\bf c} \in \mathcal{C}$ such that ${\bf c} = [c_1 \ c_2 \ \cdots \ c_n]$, the support of ${\bf c}$ is defined as $\text{supp}({\bf c}) = \{i \in [n], c_i \neq 0\}$.  We recall the well known fact that $\left(\mathcal{C}|_S\right)^{\perp} = \left(\mathcal{C}^{\perp}\right)^S$. The all-zero codeword shall be denoted by $\bf{0}$. Further, if matrix $G$ denotes a generator matrix of $\mathcal{C}$, we write $G|_S$ to denote the submatrix obtained by restricting $G$ to the columns specified by $S$. Thus $G|_S$ denotes a generator matrix for the code $\mathcal{C}|_S$. Unless otherwise specified, we only deal with linear codes in this manuscript.

\section{Background on Maximally Recoverable Codes} \label{sec:mrc}

In this section, we review relevant known facts regarding MRSCs, including equivalent definitions, existence of MRSCs and the notion of MRCs with locality.  

\vspace{0.1in}

\begin{defn}[$\ell$-cores~\cite{GopHuaSimYek}]
Consider  an $[n, k]$ code $\mathcal{C}$, and let $\mathcal{C}^{\perp}$ denote the dual of $\mathcal{C}$. Any set $S \subseteq [n], |S| = \ell$  is termed as an $\ell$-core of $\mathcal{C}^{\perp}$ if $\text{supp}\left(\bf{c}\right) \nsubseteq S, \forall \bf{c} \in \mathcal{C}^{\perp}, \bf{c} \neq \bf{0}$.
\end{defn}

\vspace{0.1in}

We next state certain equivalent definitions of MRSCs. These are mostly known from various existing works in the context of MRCs with locality, and are straightforward to verify. A proof is however included in Appendix \ref{app:MRC_eq} for the sake of completeness.

\vspace{0.1in}

\begin{lem}\label{lem:MRC_alt}
Consider an $[n, t]$ code $\mathcal{C}_0$, and let $\mathcal{C}$ denote an $[n, k]$ subcode of $\mathcal{C}$. Also let $G_0$ and $G$ denote generator matrices for $\mathcal{C}_0$ and $\mathcal{C}$, respectively. Then, the following statements are equivalent:
\begin{enumerate}[1.]
\item $\mathcal{C}$ is a maximally recoverable subcode of $\mathcal{C}_0$, as defined in Definition \ref{defn:mrsc}, i.e., for any set $S \subseteq [n], |S| = k$ such that $\rho\left(G_0|_S\right) = k$, we have $\rho\left(G|_S\right) = k$.
\item Any set $S \subseteq [n], |S| = k$ which is a $k$-core of $\mathcal{C}_0^{\perp}$ is also a $k$-core of $\mathcal{C}^{\perp}$.
\item For any set $S \subseteq [n], |S| = k$ which is a $k$-core of $\mathcal{C}_0^{\perp}$, we have $\rho(H|_{[n]\backslash S}) = n-k$. Here $H$ denotes a parity check matrix for the code $\mathcal{C}$.
\item For any set $S \subseteq [n], |S| \leq k$, $\rho\left(G_0|_S\right) =  \rho\left(G|_S\right)$.  Since the dual of a punctured code is  a shortened code of the dual code, this is equivalent to saying that $\mathcal{C}$ is an MRSC of $\mathcal{C}_0$ if and only if for any $S \subseteq [n], |S| \leq k$, $\left(\mathcal{C}^{\perp}\right)^S =  \left(\mathcal{C}_0^{\perp}\right)^S$. We note that this is further equivalent to saying that   any $k$-sparse vector ${\bf c}$ that is a codeword  of $\mathcal{C}^{\perp}$, is also a codeword of $\mathcal{C}_0^{\perp}$.
\end{enumerate}
\end{lem}
\begin{proof}
See Appendix \ref{app:MRC_eq}.
\end{proof}

\vspace{0.1in}

The following lemma, which is restatement of Lemma $14$ in \cite{GopHuaSimYek}, guarantees the existence of MRSCs under sufficiently large field size. 

\vspace{0.1in}

\begin{lem}[\cite{GopHuaSimYek}] \label{lem:mrc_exist}
Given any $[n, t]$ code $\mathcal{C}_0$ over $\mathbb{F}_q$, and any $k$ such that $k < t$, there exists an $[n, k]$ maximally recoverable subcode $\mathcal{C}$ of $\mathcal{C}_0$ , whenever $q > kn^k$.
\end{lem}

\vspace{0.1in}

We next present an explicit construction of MRSCs based on matrix representations corresponding to {\em linearized polynomials}. Linearized polynomial based constructions for MRCs with locality, from a parity-check matrix point of view, appear in \cite{Parikshit_mrc}. In the same paper~\cite{Parikshit_mrc}, the authors show, again using linearized polynomials and parity-check matrix ideas, how to obtain an explicit construction of an MRSC $\mathcal{C}$ of $\mathcal{C}_0$, when $G_0$ is a binary matrix.  The construction described below\footnote{We do not claim novelty for Construction \ref{constr:mrc_linearized}, since the idea used here follows directly from works like \cite{SilRawVis}, \cite{RawKoySilVis}, \cite{IIScUTA_MBR}} uses a generator matrix view point, and the technique is similar to the usage of linearized polynomials appearing in works of \cite{SilRawVis}, \cite{RawKoySilVis}, \cite{IIScUTA_MBR}. These works used linearized polynomials for constructing  locally repairable codes. As we will see, the construction of the {\em sandwiched } MRSCs that we present in the context of the broadcast setting (see Fig. \ref{fig:sys_model_diff_function}) is an adaptation of the following construction.

\vspace{0.1in}

\begin{constr} \label{constr:mrc_linearized}
Consider the $[n, t]$ code $\mathcal{C}_0$ over $\mathbb{F}_q$, having a generator matrix $G_0$. Let $\mathbb{F}_Q$ denote an extension field of $\mathbb{F}_q$, where $Q = q^t$. Let $\mathcal{C}_0^{(Q)}$ denote the $n$-length (linear) code over  $\mathbb{F}_Q$ that is also generated by $G_0$. Note that $\text{rank}_{\mathbb{F}_Q}(G_0) = \text{rank}_{\mathbb{F}_q}(G_0)$, and thus  $\text{dim}\left(\mathcal{C}_0^{(Q)} \right) = t$. In this construction, we identify an $[n, k], k \leq t$ MRSC of the $[n, t]$ code $\mathcal{C}_0^{(Q)}$. Toward this, let $\{\alpha_i \in \mathbb{F}_Q, 1 \leq i \leq t\}$ denote a basis of $\mathbb{F}_Q$ over $\mathbb{F}_q$. Define the elements $\{\beta_i \in \mathbb{F}_Q, 1 \leq i \leq n\}$ as follows:
\begin{eqnarray} \label{eq:revise_1}
[\beta_1 \ \beta_2 \ \cdots \ \beta_n] & = & [\alpha_1 \ \alpha_2 \ \cdots \ \alpha_t]G_0.
\end{eqnarray}
Next, consider the $n$-length code $\mathcal{C}$ over $\mathbb{F}_Q$ having a generator matrix $G$ given by 
\begin{eqnarray}
G & = & \left[  \begin{array}{cccc} \beta_1 & \beta_2 & \cdots & \beta_n   \\
\beta_1^q & \beta_2^q & \cdots & \beta_n^q \\
& & \vdots &  \\
\beta_1^{q^{k-1}} & \beta_2^{q^{k-1}} & \cdots & \beta_n^{q^{k-1}} \end{array} \right].
\end{eqnarray}
The code $\mathcal{C}$ is our candidate code, and this completes the description of the construction. In the following theorem, we prove that $\mathcal{C}$ is indeed an $[n, k]$ MRSC of $\mathcal{C}_0^{(Q)}$.
\end{constr}	

\vspace{0.1in}

\begin{thm} \label{thm:mrc_linearized}
Consider an $[n, t]$ code $\mathcal{C}_0$ over $\mathbb{F}_q$, having a generator matrix $G_0$, where $G_0 \in \mathbb{F}_q^{t \times n}$. Let $\mathcal{C}_0^{(Q)}$ denote the $n$-length (linear) code  over  $\mathbb{F}_Q$ that is also generated by $G_0$, where $Q = q^t$. Then the code $\mathcal{C}$ over $\mathbb{F}_Q$ obtained in Construction \ref{constr:mrc_linearized} is an $[n, k]$ maximally recoverable subcode of $\mathcal{C}_0^{(Q)}$.	
\end{thm}
\begin{proof}
See Appendix \ref{app:mrc_linearized}.
\end{proof}	

\subsection{Maximally Recoverable Codes with Locality (Partial MDS Codes)}

Below we give the definition of MRCs with locality, using the notion of MRSCs presented in Definition \ref{defn:mrsc}.

\vspace{0.1in}

\begin{defn}[MRCs with Locality~\cite{Parikshit_mrc},~\cite{GopHuaSimYek}] \label{defn:partial_mds}
Assume that the parity matrix $H_0$ of $\mathcal{C}_0$ has the following form:
\begin{eqnarray} \label{eq:H0}
H_0 & = & \left[ \begin{array}{cccc} H_{L,1}&&& \\& H_{L,2} && \\ && \ddots & \\ &&&H_{L,\ell}  \end{array} \right],
\end{eqnarray}
where each $H_{L, i}, 1 \leq i \leq \ell$ generates an $[r+\delta-1, \delta-1]$ MDS code, for some fixed $r, \delta$. The code $\mathcal{C}_0^{\perp}$ has parameters $[n = \ell(r+\delta-1), n-t = \ell(\delta-1)]$. Then an $[n, k]$ MRSC $\mathcal{C}$ of $\mathcal{C}_0$ will be referred to as an $[n, k]$ $(r, \delta)$ maximally recoverable code with locality.
\end{defn}

\vspace{0.1in}

We would like to point out that the precise form of $H_{L,i}, 1  \leq i \leq \ell$ is not imposed by the above definition. In other words, while designing an  $[n, k]$ $(r, \delta)$ MRC with locality, we are free to choose any $H_{L,i}, 1  \leq i \leq \ell$ that generates an $[r+\delta-1, r, \delta]$ MDS code. MRCs with locality are also known in literature as partial MDS codes~\cite{blaum_mds}. An $[n, k]$ $(r, \delta)$ MRC with locality corresponds to an $[m', n'] (r',s')$ partial MDS code~\footnote{Notation used for partial MDS codes correspond to the one used in \cite{blaum_mds}.}, where 
\begin{eqnarray}
m' & = & \frac{n}{r+\delta-1} \\
n' & = & r+\delta - 1 \\
r' & = & \delta - 1\\
s' & = &  \frac{n}{r+\delta-1}r - k.
\end{eqnarray}

The idea here is that we see the code as a two dimensional array code having $m'$ rows and $n'$ columns. Each row of the array is an $[n' = r+\delta - 1, n' - r' = r]$ MDS code having $r' = \delta-1$ local parities. The parameter $s' = \text{dim}(\mathcal{C}_0) - \text{dim}(\mathcal{C})$ refers to the number of global parity symbols.  An $[m', n'] (r',s')$ partial MDS code can tolerate any combination of $r'$ erasures  in each row, and an additional $s'$ erasures among the remaining elements in the array. In this write-up we will follow the terminology of MRCs with locality.

\section{Communication Cost for the Point-to-Point Case}\label{sec:p2p}

In this section, we obtain necessary and sufficient conditions for solving the point-to-point problem setting in Fig. \ref{fig:sys_model}. An encoder-decoder pair $(H, \mathcal{D})$ will be referred to as a {\em valid scheme} for the problem in Fig. \ref{fig:sys_model}, if the decoder's estimate of $A(X^n  + E^n)$ is correct for all $X^n, E^n \in \mathbb{F}_q^n$, such that $\text{Hamming wt}.(E^n) \leq \epsilon$. Recall our assumption that we  deal with a zero probability of error, worst-case scenario model. Also, recall that the communication cost associated with the encoder $H$ is given by $\ell = \rho(H)$, where the parameters $n, q, \epsilon$ and $A$ of the system are assumed to be fixed a priori. Without loss of generality, we assume that the $m \times n$ matrix $A$ has rank $m$.

\subsection{Necessary Conditions}

\begin{thm}\label{thm:necessary}
Consider any valid scheme $(H, \mathcal{D})$ for the communication problem described in Fig. \ref{fig:sys_model}. Let $\mathcal{C}_H$ and $\mathcal{C}_A$ denote the linear codes generated by the rows of the matrices $H$ and $A$ respectively. Also, consider the code $\mathcal{C} = \mathcal{C}_H \cap \mathcal{C}_A$, and let $P$ denote a generator matrix for $\mathcal{C}$. Then, the following conditions must necessarily be satisfied:
\begin{enumerate}[1.]  
\item If $Y^n$ is any $2\epsilon$-sparse vector such that $AY^n \neq {\bf 0}$, then it must also be true that $PY^n \neq {\bf 0}$.
\item $\text{dim}\left( \mathcal{C} \right) \geq \min(m, 2\epsilon)$.  
\end{enumerate}
Combining the two parts it then follows that if $\text{dim}\left( \mathcal{C} \right) = \min(m, 2\epsilon)$, then $\mathcal{C}$ must be a maximally recoverable subcode of $\mathcal{C}_A$.
\end{thm}
\begin{proof}
$1.$ We will prove the first part by contradiction. Let us suppose that there exists a non-zero $2\epsilon$-sparse vector $Y^n$ such that $Y^n \in \mathcal{C}^{\perp}$, but $Y^n \notin \mathcal{C}_A^{\perp}$. In this case, we will show that there exists two distinct pairs $(X_1^n, E_1^n)$ and $(X_2^n, E_2^n)$, with $E_1^n$ and $E_2^n$ being $\epsilon$-sparse, such that $H(X_1^n + E_1^n) = H(X_2^n + E_2^n)$ and $AX_1^n = AX_2^n$, but $AE_1^n \neq AE_2^n$. Thus, no decoder $\mathcal{D}$ can resolve between the two pairs successfully, and we would have contradicted the validity of our scheme. Towards constructing the desired $(X^n, E^n)$ pairs, note that $Y^n$ can be expressed as $Y^n = E_1^n - E_2^n$ for some two distinct $\epsilon$-sparse vectors $E_1^n$ and $E_2^n$.
Also, since $\mathcal{C}^{\perp} = \mathcal{C}_H^{\perp} + \mathcal{C}_A^{\perp}$, there exists $U^n \in \mathcal{C}_A^{\perp}$ and $V^n \in \mathcal{C}_H^{\perp}$ such that $E_1^n - E_2^n = U^n + V^n$. We now choose $X_1^n = -U^n$ and $X_2^n = {\bf 0}$. In this case, we see that $H(X_1^n + E_1^n) = H(X_2^n + E_2^n) = HE_2^n $ and $AX_1^n = AX_2^n  = {\bf 0}$, but $AE_1^n \neq AE_2^n$, which contradicts the validity of our scheme. This completes the proof of the first part.

$2.$ Assume that $\text{dim}(\mathcal{C}) < \min(m, 2\epsilon)$ which implies that $\text{dim}\left(\mathcal{C}^{\perp}\right) > n - \min(m, 2\epsilon)$. In this case,  we note that there exists a basis $\mathcal{L}$ of $\mathcal{C}^{\perp}$ consisting entirely of vectors that are $2\epsilon$-sparse. To see why this is true, take any basis of $\mathcal{C}^{\perp}$ and row-reduce it to the standard form (up to permutation of columns) given by  $[I\ | \ K]$, where $I$ denotes the identity matrix of size $\text{dim}(\mathcal{C}^{\perp})$. Since $\text{dim}(\mathcal{C}^{\perp}) > n - 2\epsilon$, the number of columns in the matrix $K$ is at most $2\epsilon - 1$, and thus every row in matrix $[I\ | \ K]$ is $2\epsilon$-sparse, which proves the existence of the basis $\mathcal{L}$ as stated above. Next, observe that  $\text{dim}\left(\mathcal{C}\right) < \min(m, 2\epsilon) \leq \text{dim}\left(\mathcal{C}_A\right)$, which implies that $\text{dim}\left(\mathcal{C}^{\perp}\right) >  \text{dim}\left(\mathcal{C}_A^{\perp}\right)$. Thus at least one of the elements of the set $\mathcal{L}$ is not contained in $\mathcal{C}_A^{\perp}$. In other words,  there exists a non-zero $2\epsilon$-sparse vector $Y^n$ such that $Y^n \in \mathcal{C}^{\perp}$, but $Y^n \notin \mathcal{C}_A^{\perp}$. However, this would contradict Part $1.$ of the theorem that we just proved above. We thus conclude that $\text{dim}(C) \geq \min(m, 2\epsilon)$.

Finally, suppose that $\text{dim}\left( \mathcal{C} \right) = \min(m, 2\epsilon)$. If $m \leq 2\epsilon$, then $\mathcal{C} = \mathcal{C}_A \cap \mathcal{C}_H = \mathcal{C}_A$, since by assumption $\text{dim}\left( \mathcal{C}_A \right) = m$. In this case, it is trivially true that $\mathcal{C}$ is a MRSC of $\mathcal{C}_A$. Now if $m > 2\epsilon$, we note that Part $1$ implies that if $S \subseteq [n], |S| = 2\epsilon$ is such that $\rho(A|_S) = 2\epsilon$, then $\rho(P|_S) = 2\epsilon$. It then follows from Definition \ref{defn:mrsc} that $\mathcal{C}$ must be a MRSC of $\mathcal{C}_A$ in this case as well, if $\text{dim}\left( \mathcal{C} \right) = 2\epsilon$.
\end{proof}

\vspace{0.1in}

\begin{cor}\label{cor:necessary}
\begin{enumerate}[a)]
\item The communication cost $\ell$ associated with any valid scheme for the problem given in Fig. \ref{fig:sys_model} is lower bounded by $\ell  \geq  \min(m, 2\epsilon)$.
\item If $m \leq 2\epsilon$ and $\ell = m$ (optimal communication cost), then the encoding matrix $H$ is necessarily given by $H = A$.
\item If $m > 2\epsilon$ and $\ell = 2\epsilon$ (once again optimal communication cost), then it must necessarily be true that the code $\mathcal{C}_H$ is a MRSC of $\mathcal{C}_A$. 
\end{enumerate}
\end{cor}

\vspace{0.1in}

\begin{note} \label{note:Shah}
	The bound $\ell  \geq  \min(m, 2\epsilon)$ in Part $a)$ of the above corollary  was already proved in \cite{Shah_oblivious} for an average case setting, under the assumption that the vectors $X^n$ and $X^n + E^n$ are uniformly picked at random from the set of all vectors satisfying  $\text{Hamming wt}.(E^n) \leq \epsilon$. This result in \cite{Shah_oblivious} can be directly used to argue the correctness of Part $a)$ of the above corollary, for the worse case setting that we consider here. However, the connection to maximally recoverable codes, especially the necessity of MRSCs for achieving optimality (Part $c)$ of  the above corollary) is a novelty of this paper, and as we will see next forms the basis of the our achievability scheme for an arbitray matrix $A$ chosen for a sufficiently large field size.
\end{note}

\subsection{Optimal Achievable Scheme based on Maximally Recoverable Subcodes, when $m > 2\epsilon$} \label{sec:achievability}

We now present a valid scheme for the case $m > 2\epsilon$, having optimal communication cost $\ell = 2\epsilon$.  From Corollary \ref{cor:necessary} , we know that the $[n, 2\epsilon]$ code $\mathcal{C}_H$  must necessarily be an MRSC of $\mathcal{C}_A$. We also know from Lemma \ref{lem:mrc_exist} that such an MRSC always exists whenever the field size $q > 2\epsilon n^{2\epsilon}$.  Since $\mathcal{C}_H$ is a subcode of $\mathcal{C}_A$, there exists a $2\epsilon \times m$ matrix $S$ such that $H = SA$. Given the matrices $H$ and $S$, we now refer to Fig. \ref{fig:achievability} for a schematic of the decoder to be used. The various steps performed by the decoder to estimate $A(X^n + E^n)$ are as follows:

\begin{enumerate}[(1)]
\item Given the encoder output $H(X^n + E^n)$ and the side information $AX^n$, obtain $HE^n$ as 
\begin{eqnarray}
HE^n & = & H(X^n + E^n) - S\left(AX^n\right).
\end{eqnarray}
\item Determine any $\epsilon$-sparse vector $\widehat{E}^{n}$ such that $HE^n = H\widehat{E}^{n}$.  Since the vector $\left(E^n - \widehat{E}^{n}\right)$ is $2\epsilon$-sparse, and since $\mathcal{C}_H$ is a $2\epsilon$-dimensional MRSC of $\mathcal{C}_A$, we know from Part $4$. of Lemma \ref{lem:MRC_alt} that $HE^n = H\widehat{E}^{n} \implies AE^n = A\widehat{E}^{n}$.
\item Compute the desired estimate as $AX^n + A\widehat{E}^{n} = A(X^n + E^n)$. 
\end{enumerate}

\begin{figure*}[h]
  \centering
\includegraphics[width=6.5in]{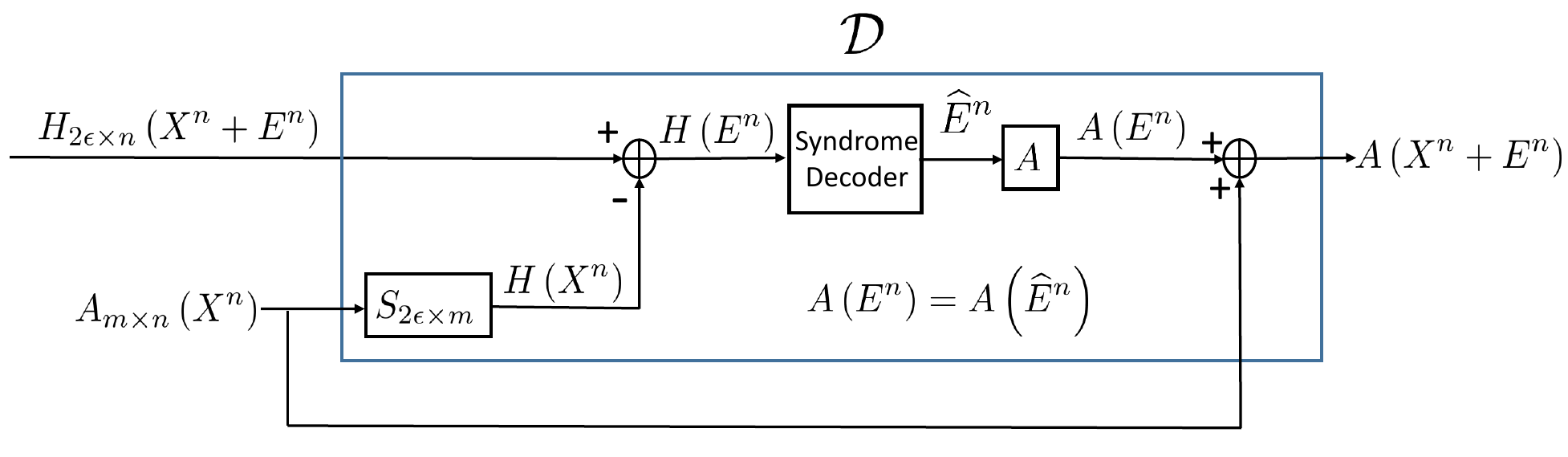}
  \caption{The decoder used in the achievability proof for the case $m > 2\epsilon$.}
  \label{fig:achievability}
\end{figure*}

\vspace{0.1in}

\begin{note} 
We note that it is possible to construct a valid scheme, not necessarily optimal, using any matrix $H$ which satisfies the conditions of Theorem \ref{thm:necessary}. This also implies that if $(H, \mathcal{D})$ is any valid scheme for the problem, and if $\mathcal{C}_H \not \subset \mathcal{C}_A$, then a  second valid scheme for the problem having lower communication cost can be constructed using the matrix $P$ for encoding, where $P$ is a generator matrix for  the code $\mathcal{C}_A \cap \mathcal{C}_H$. This follows because the matrix $P$ will also satisfy the conditions of the Theorem \ref{thm:necessary}. 
\end{note}

\vspace{0.1in}

\begin{note} [Decoding Complexity]\label{rem:decode_complex}
	We overlook the issue of decoding complexity in the above achievable scheme. The main purpose of the above scheme is two fold: $1)$ To establish the optimality of the lower bound on communication cost, provided in Corollary \ref{cor:necessary}, and $2)$ to show how MRSCs are useful (and necessary) in constructing an achievable scheme. The above scheme works for any arbitrary matrix $A$, as long as the field size is sufficiently large. In Fig. \ref{fig:achievability}, we assume syndrome decoding to determine the $\epsilon$-sparse vector $E^n$. For small values of $\epsilon$, which is the case if the updates are tiny, syndrome decoding can be efficiently implemented via look-up table. For large $\epsilon$, which occurs for instance, when the updates are a constant fraction of the size of the file, the decoding complexity is indeed high for the above scheme. 
\end{note}

\vspace{0.1in} 

\begin{note} \label{rem:using_lin}
In this remark, we comment about the usability of the explicit construction \ref{constr:mrc_linearized} for our achievable scheme. The above achievability result is based on MRSCs obtained via Lemma \ref{lem:mrc_exist}. It is a natural question as to whether one could instead use the explicit MRSCs obtained via Construction \ref{constr:mrc_linearized}. The motivation for this question arises from the issue of decoding complexity discussed in Remark \ref{rem:decode_complex}. Recall that Construction \ref{constr:mrc_linearized} is based on linearized polynomials, and these are 
closely connected to rank-metric codes (see Problem $11.4$, \cite{roth} for a quick definition, and the connection). Decoding complexity of rank-metric codes is a well studied topic \cite{loidreau2006welch, richter2004error, gabidulin1991fast, roth1991maximum, gadouleau2008complexity}, and efficient algorithms are known. Given this existing literature,  it is of interest to find out if one can use Construction \ref{constr:mrc_linearized} in our achievability proof, so that when the matrix $A$ is also structured (like in Examples \ref{ex:moti}, \ref{ex:bcast_use}), one might hope to get efficient decoding algorithms.

The short answer is to the above question is that we cannot directly use Construction \ref{constr:mrc_linearized} without a change in the system model. To explain this, consider the Construction \ref{constr:mrc_linearized}, and note that given the matrix $A$ over finite field $\mathbb{F}_q$, the construction obtains the MRSC over the extension field $\mathbb{F}_Q$, where $Q = q^m$, where $m = rank(A)$. A large $q$ does not help (like in the existence proof); the construction extends the field irrespective of what $q$ is. If we use this construction directly in the achievability scheme, the overall cost gets multiplied by a factor of $t$ and the communication cost is $2t\epsilon$ instead of $2\epsilon$. 
	
One possible workaround, in order to avoid taking a hit on the communication cost by a factor of $t$, is to change the system model as follows. Instead of encoding one vector $X^n$, we encode $t$ vectors $X_1^n, \ldots X_t^n$, where we assume that the receiver is interested in the $t$ linear combinations $AX_1^n, \ldots, AX_t^n$. In other words, the input to the encoder is now a matrix over $\mathbb{F}_q$ instead of a vector. Now, each row of the matrix is replaced with a single element in the extended field $\mathbb{F}_Q$, and let $X_{(Q)}^n$ denote 	new encoding vector over $\mathbb{F}_Q$. The receiver is interested in decoding $AX_{(Q)}^n$, where $A$ is as before, but can be considered as matrix over $\mathbb{F}_Q$ itself. With this model change, one can use Construction \ref{constr:mrc_linearized} to find an MRSC of $A$, and use that to encode $X_{(Q)}^n$, decode $AX_{(Q)}^n$ and finally recover $AX_1^n, \ldots, AX_t^n$ from $AX_{(Q)}^n$. The last step works since the elements of $A$ come from $\mathbb{F}_q$. In this modified set up, one can explore the utility of efficient decoding algorithms for rank-metric codes for the decoding problem at hand. Finally, we note that in the modified scheme the parameter $\epsilon$ denotes the sparsity of $X_{(Q)}^n$.
\end{note}	

\subsection{Bound on Communication cost with non-linear codes and  Optimality of linear encoding} \label{sec:nonlin}

In this section, we present a lower bound on the communication cost for the point-to-point case allowing for non-linear encoding. Recall that the bound in Corollary \ref{cor:necessary} assumed linear encoding. We show the interesting result that even with non-linear encoding the communication cost is lower bounded by the $\ell \geq \min(m, 2\epsilon)$, which is the same bound that we have in Corollary \ref{cor:necessary}. This shows that whenever a linear encoder exists that achieves cost of $\ell = \min(m, 2\epsilon)$, the linear encoder is optimum among the class of all encoders. Recall from the achievability proof that optimal linear encoder exists for any arbitrary matrix $A$, as long as the field size is sufficiently large. We note that while the bound obtained in this section subsumes the bound of Corollary \ref{cor:necessary}, the analysis here does not establish the necessity of MRSCs while restricting to linear encoders. For this reason, we choose to keep the discussion on non-linear encoding separate.

Our proof technique is based on finding a lower bound on the chromatic number of the characteristic graph, also known as confusability graph, associated with the function computation problem. Characteristic graphs have been used in the past for studying 
zero-error source compression problems~\cite{korner1998zero, doshi2010functional}. In \cite{korner1998zero}, the authors show how chromatic number of the characteristic graph appears as part of a lower bound for the communication cost when communicating a source $X$ to a distant receiver that has access to side information $Y$ (correlated with $X$), and is interested in recovering $X$. In our work, we rely on a result from \cite{doshi2010functional} that extends the above result to the setting when the receiver is interested in computing, with zero-error, a function of $X$ and $Y$. For sake of completeness, we provide a self-contained description of the analysis. 

\vspace{0.1in}

\begin{defn}[Characteristic Graph]
Consider the function computation problem defined in Fig. \ref{fig:sys_model}. Let $V = \{\hat{X}^n\}$ denote the set of all possible inputs to the encoder. In other words, $V$ is simply the collection of all possible vectors in $\mathbb{F}_q^n$.  Consider the graph $\mathcal{G}$ having $V$ as the set of vertices. Let $Y^n = AX^n$ denote the side-information available at the decoder. If $\hat{X}^n$ is the input, we know that $\text{Hamming wt.}(\hat{X}^n - X^n) \leq \epsilon$. In this case, we say that the pair $(\hat{X}^n, Y^n)$ can jointly occur. An edge exists between vertices corresponding to $\hat{X}_1^n$ and $\hat{X}_2^n$, if $A\hat{X}_1^n \neq A\hat{X}_2^n$ and if there is a $Y^n = AX^n$, for some $X^n$ such that both the pairs $(\hat{X}_1^n, Y^n)$ and $(\hat{X}_2^n, Y^n)$ can jointly occur.  The graph $\mathcal{G}$ is defined as the characteristic graph associated with the function computation problem. 
\end{defn}

\vspace{0.1in}

The graph is defined such that an edge exists between vertices $\hat{X}_1^n$ and $\hat{X}_2^n$ if $\hat{X}_1^n$ can be \textit{confused} for $\hat{X}_2^n$. For this reason, the graph is also referred to as the confusability graph. A coloring of $\mathcal{G}$ assigns colors to the its vertices such that any pair of neighboring vertices are assigned distinct colors. Two vertices are neighbors if there is a edge between them. The chromatic number, $\chi(\mathcal{G})$, is the minimum number of colors present in any coloring of $\mathcal{G}$. 

Let $f$ denote an encoder for the function computation problem in Fig. \ref{fig:sys_model}, where we allow $f$ to be non-linear. We define the notion of valid scheme for the encoder-decoder pair $(f, \mathcal{D})$ like we did for the case of linear encoders. Thus, an encoder-decoder pair $(f, \mathcal{D})$ will be referred to as a {\em valid scheme} for the problem in Fig. \ref{fig:sys_model} (with the encoder $H$ replaced with $f$), if the decoder's estimate of $A(X^n  + E^n)$ is correct for all $X^n, E^n \in \mathbb{F}_q^n$, such that $\text{Hamming wt}.(E^n) \leq \epsilon$. The following lemma is a paraphrasing of Lemma $26$ of \cite{doshi2010functional}.

\vspace{0.1in}

\begin{lem}
Any valid scheme $(f, \mathcal{D})$ for the function computation problem in Fig. \ref{fig:sys_model}(with the encoder $H$ replaced with $f$) induces a coloring of the characteristic graph $\mathcal{G}$ such that the number of colors in the coloring is given by $|Range(f)|$. Here $|Range(f)|$ is the cardinality of the set of all possible encoder outputs.
\end{lem}
\begin{proof}
Lets us assign $f(\hat{X}^n)$ as the color of vertex $\hat{X}^n$. We claim that this assignment is a coloring of $\mathcal{G}$. To prove this, let $\hat{X}_1^n$ and $\hat{X}_2^n$ be neighbors in $\mathcal{G}$. From the definition of $\mathcal{G}$, this means that $A\hat{X}_1^n \neq A\hat{X}_2^n$ and there is a $Y^n = AX^n$, for some $X^n$ such that both the pairs $(\hat{X}_1^n, Y^n)$ and $(\hat{X}_2^n, Y^n)$ can jointly occur. In this case, since $(f, \mathcal{D})$ is a valid scheme, it must necessarily be true that $f(\hat{X}_1^n) \neq f(\hat{X}_1^n)$; otherwise it is impossible for the decoder $\mathcal{D}$ to distinguish between $\hat{X}_1^n$ and $\hat{X}_2^n$, when $Y^n$ appears as side-information. This proves that the assignment $f(\hat{X}^n)$ as the color of vertex $\hat{X}^n$ results in a coloring of $\mathcal{G}$.
\end{proof}	

\vspace{0.1in}
 
\begin{cor} \label{cor:chrom}
For any valid scheme $(f, \mathcal{D})$ for the function computation problem in Fig. \ref{fig:sys_model}(with the encoder $H$ replaced with $f$), we have 	$|Range(f)| \geq \chi(\mathcal{G})$, where $\mathcal{G}$ is the characteristic graph associated with the function computation problem.
\end{cor}

\vspace{0.1in}

In the following lemma, we provide a lower bound on the chromatic number $\chi(\mathcal{G})$. We assume that $m \geq 2\epsilon$. The case when $m < 2\epsilon$ can be similarly handled.

\begin{lem} \label{lem:chrom}
The chromatic number $\chi(\mathcal{G})$ of the characteristic graph $G$ associated with function computation problem in Fig. \ref{fig:sys_model} is lower bounded by $\chi(\mathcal{G}) \geq q^{2\epsilon}$, whenever $m = \rho(A) \geq 2\epsilon$.
\end{lem}	
\begin{proof}
Without loss of generality, assume that $\rho(A|_{[2\epsilon]}) = 2\epsilon$. Let $\hat{X}_1^n$ and $\hat{X}_2^n$ be such that $\text{supp}(\hat{X}_1^n) \subseteq [2\epsilon], \text{supp}(\hat{X}_2^n) \subseteq [2\epsilon]$.  Clearly, if $\hat{X}_1^n \neq \hat{X}_2^n$, it must be true that $A\hat{X}_1^n \neq A\hat{X}_2^n$, otherwise we have $A(\hat{X}_1^n - \hat{X}_2^n) = 0$, which violates the assumption that  $\rho(A|_{[2\epsilon]}) = 2\epsilon$. Further, since both $\hat{X}_1^n$ and $\hat{X}_2^n$ have their supports limited to the first $2\epsilon$ coordinates, one can find a vector $X^n$ such that $\text{Hamming wt.}(\hat{X}_1^n - X^n) \leq \epsilon$ and  $\text{Hamming wt.}(\hat{X}_2^n - X^n) \leq \epsilon$. Now consider the funciton computation probem where the side information $Y^n$ is given by $Y^n = AX^n$, where $X^n$ is as above. Thus the both pairs  $(\hat{X}_1^n, Y^n)$ and $(\hat{X}_2^n, Y^n)$ can jointly occur. 	Combining the above observations, we conclude that $\hat{X}_1^n$ and $\hat{X}_2^n$ must be neighbors in the characteristic graph $G$ . This then implies that if we consdier the subgraph of $\mathcal{G}$ obtained by restricting $\mathcal{G}$ to the set of $q^{2\epsilon}$ vertices corresponding to $q^{2\epsilon}$ vectors whose support is confined to $[2\epsilon]$, this subgraph is a fully-connected graph (clique). From this we conclude that $\chi(\mathcal{G}) \geq q^{2\epsilon}$, whenever $m \geq 2\epsilon$.
\end{proof}	

\vspace{0.1in}

\begin{thm}
The communication cost of associated with any valid scheme $(f, \mathcal{D})$ for the function computation problem in Fig. \ref{fig:sys_model}(with the encoder $H$ replaced with $f$), is lower bounded by $\ell \geq 2\epsilon$,  whenever $m = \rho(A) \geq 2\epsilon$.
\end{thm}
\begin{proof}
Communication cost is given by $\log_q(|Range(f)|)$. The proof now follows by combining Corollary \ref{cor:chrom} with Lemma \ref{lem:chrom}.
\end{proof}

\vspace{0.1in}

\begin{note}
The determination of chromatic number for a general graph is an NP-complete problem. In the above analysis, we  obtained lower bound on the chromatic number for the class of characteristic graphs associated with the function computation problem. It is interesting to note that our achievability proof based on MRSCs provide exact characterization of the chromatic number for the characteristic graph, whenever the field size is sufficiently large. It is unclear to us whether there is an alternate direct method to find chromatic number of the characteristic graph. If such a method exists, it might give further insights on the field size requirement of MRSCs.
\end{note}

\section{Updating Linearly Encoded Striped-Data-Files} \label{sec:p2p-example}

In this section, we present a low-field MRSC construction for the setting of Example \ref{ex:moti}, where we allow striping and encoding by an arbitrary $[N, K]$ linear code $\mathcal{C}_{st}$ over $\mathbb{F}_q$. Recall that in Example \ref{ex:moti}, we apply the point-to-point model in Fig. \ref{fig:sys_model} in order to update the coded data in any one of the $N$ storage nodes in the system. The $K$-length coding vector for one of the storage nodes is given by ${\bf a}  = [a_1 \ a_2 \ \cdots \ a_K], a_i \in \mathbb{F}_q$. The vector $X^n$ denotes the uncoded data file, which is divided into $m$ stripes each of length $K$ symbols. The first stripe consists of the first $K$ symbols of $X^n$, the second stripe consists of symbols ${X_{K+1}, \ldots, X_{2k}}$, and so on. We assume that the length $n = mK$. The overall coding matrix $A$, which corresponds to the desired function at the destination, is given by 
\begin{eqnarray} \label{eq:Aspecial}
A & = & \left[ \begin{array}{cccc}
{\bf a} & & & \\ & {\bf a} & &  \\ & & \ddots & \\ & & & {\bf a}
\end{array} \right].
\end{eqnarray} 
If we design the optimal encoder $H$ based on existence of MRSCs from Lemma \ref{lem:mrc_exist}, then the field size $q$ must be at least $2\epsilon n^{2\epsilon}$. 

We now present an alternate low field size explicit construction of MRSCs of $\mathcal{C}_A$, when the matrix $A$ takes on the special form given in \eqref{eq:Aspecial}. For the coding vector ${\bf a}  = [a_1 \ a_2 \ \cdots \ a_K]$, if assume that $a_i \neq 0, \forall i \in [K]$,  then note that an $[n, 2\epsilon]$ MRSC $\mathcal{C}_H$ of $\mathcal{C}_A$ is in fact an $[n, 2\epsilon] (K, 1)$ MRC with locality (see Definition \ref{defn:partial_mds}), i.e., a code with locality where all the local codes are scaled repetition codes. The problem of low field size constructions of MRCS with locality is in general an open problem. However, as we show here, for the special case when the local codes appear as repetition codes, low field size constructions are easily identified.

We divide the discussion into two parts. We will first show how MRSCs of any general code $\mathcal{C}_0$ can be obtained by first suitably {\em extending} $\mathcal{C}_0$, and then shortening the extended code. Given this result, we then show that the problem of constructing an MRSC of $\mathcal{C}_A$ where $A$ is as given by \eqref{eq:Aspecial}, is as simple as finding an $[m, m-2\epsilon]$ MDS code over $\mathbb{F}_q$. If we employ Reed-Solomon codes, a field size $q > m$ is sufficient for the construction.

\subsection{An Equivalent Definition of MRSCs via Code Extensions} \label{sec:mrsc_extension}

Consider an $[n, t]$ code $\mathcal{C}_0$ over $\mathbb{F}_q$ having generator matrix $G_0$. Define a new $[n + \Delta, t]$ code $\mathcal{C}_0^{(e)}$ over $\mathbb{F}_q$ whose generator matrix $G_0^{(e)}$ is given by
\begin{eqnarray}
G_0^{(e)} & = & [G_0 \ | \ Q],
\end{eqnarray}
for some $t \times \Delta$ matrix $Q$ over $\mathbb{F}_q$. Code $\mathcal{C}_0^{(e)}$ will be referred to as extension of $\mathcal{C}_0$. In the following lemma, we show that MRSCs of $\mathcal{C}_0$ exist if and only if 
certain extensions of $\mathcal{C}_0$ exist.

\vspace{0.1in}

\begin{lem} \label{lem:MRSC_via_shortening}
Consider an $[n + \Delta, t]$ extension $\mathcal{C}_0^{(e)}$ of the $[n, t]$ code $\mathcal{C}_0$ such that $(i)$ $\Delta < t$, and $(ii)$ for any $S \subseteq [n], |S| = t - \Delta$, we have 
\begin{eqnarray} \label{eq:extend_property}
\rho\left(G_0^{(e)}|_{S \cup \{n+1, \ldots, n+\Delta\}}\right) & = & \rho\left(G_0|_{S}\right) + \Delta.
\end{eqnarray}
Then the shortened code $\left(\mathcal{C}_0^{(e)}\right)^{[n]}$ is an $[n, k = t - \Delta]$ MRSC of $\mathcal{C}_0$. Conversely, suppose that there exists an $[n, k]$ MRSC of the $[n, t]$ code $\mathcal{C}_0$, where $k = t - \Delta$ for some $\Delta < t$. Then, there exists an $[n + \Delta, t]$ extension $\mathcal{C}_0^{(e)}$ of the $[n, t]$ code $\mathcal{C}_0$ that satisfies \eqref{eq:extend_property} for any $S \subseteq [n], |S| = t - \Delta$.
\end{lem}
\begin{proof}
We only prove the forward part here, a proof of converse appears in Appendix \ref{app:converse_MRC_extension}. To prove the forward part, we assume the existence of extension $\mathcal{C}_0^{(e)}$ satisfying \eqref{eq:extend_property}, and show that the shortened code $\mathcal{C} \triangleq \left(\mathcal{C}_0^{(e)}\right)^{[n]}$ is an $[n, k = t - \Delta]$ MRSC of $\mathcal{C}_0$.  From Part $4$, Lemma \ref{lem:MRC_alt}, we know it is sufficient to prove that 
\begin{eqnarray}
\left(\mathcal{C}^{\perp}\right)^{S} & = & \left(\mathcal{C}_0^{\perp}\right)^{S} \ \forall \ S \subseteq [n],  |S| \leq k. \label{eq:temp_shortening}
\end{eqnarray}	
Let us suppose \eqref{eq:temp_shortening} is not true; i.e., there exist a set $S^* \subseteq [n],  |S^*| \leq k$ such that  $\left(\mathcal{C}_0^{\perp}\right)^{S^*} \subsetneq \left(\mathcal{C}^{\perp}\right)^{S^*}$. In other words, there exists a vector ${\bf c} \in \mathcal{C}^{\perp}$ such that ${\bf c} \notin \mathcal{C}_0^{\perp}$ and $\text{supp}({\bf c}) \subseteq S^*$. Next, note that if $H_0$ denotes a parity check matrix for $\mathcal{C}_0$, then parity check matrices $H_0^{(e)}$ and $H$ for the codes $\mathcal{C}_0^{(e)}$ and $\mathcal{C}$ can be respectively given by 
\begin{eqnarray}
H_0^{(e)} \ = \ \left[ \begin{array}{cc} H_0 & 0_{n-t \times \Delta} \\ H_e & I_{\Delta}  \end{array} \right] & \text{and} & H = \left[ \begin{array}{c} H_0 \\  H_e \end{array} \right],
\end{eqnarray}
where $H_e$ is some $\Delta \times n$ matrix over $\mathbb{F}_q$, and $I_{\Delta}$ denotes the $\Delta \times \Delta$ identity matrix. The fact that the parity-check matrix $H$ for the code $\mathcal{C}$ is as given above, follows from the fact the dual of a shorten code is simply the punctured code of the dual code. In this case, corresponding to the vector ${\bf c}$, there exists a vector ${\bf c}^{(e)} \in \left(\mathcal{C}_0^{(e)}\right)^{\perp}$ such that 
\begin{itemize}
	\item $\text{supp}\left({\bf c}^{(e)}\right) \subseteq S^{*} \cup \{n + 1, \cdots, n + \Delta\} $,
	\item ${\bf c} = \left({\bf c}^{(e)}\right)|_{S^{*}}$, and
	\item the set $\text{supp}\left({\bf c}^{(e)}\right) \cap \{n + 1, \cdots, n + \Delta\}$ is not empty. The last part follows from our assumption that ${\bf c} \notin \mathcal{C}_0^{\perp}$.
\end{itemize}
The presence of the vector ${\bf c}^{(e)} \in \left(\mathcal{C}_0^{(e)}\right)^{\perp}$ satisfying the above conditions means that 
\begin{eqnarray} 
\rho\left(G_0^{(e)}|_{\mathcal{S^*} \cup \{n+1, \ldots, n+\Delta\}}\right) & < & \rho\left(G_0|_{\mathcal{S^*}}\right) + \Delta,
\end{eqnarray}
which is a contradiction to our assumption that the extension $\mathcal{C}_0^{(e)}$ satisfies \eqref{eq:extend_property}. From this we conclude that \eqref{eq:temp_shortening} is indeed true, and this completes the proof of the forward part.
\end{proof}

\subsection{Low Field Size MRSCs for Updating Striped-Data-Files} 

\begin{constr} \label{constr:mrc_Aspecial}
Consider the matrix $A \in \mathbb{F}_q^{m \times n}$ as given in \eqref{eq:Aspecial}, and the associated $[n, m]$ linear code $\mathcal{C}_A$ over $\mathbb{F}_q$. Next consider the $[n + m - 2\epsilon, m]$ extended code $\mathcal{C}_A^{(e)}$ over $\mathbb{F}_q$, whose generator matrix $A^{(e)}$ is given by
\begin{eqnarray} \label{eq:Aextended}
A^{(e)} & = & \left[ A \ | \ Q\right],
\end{eqnarray}
where $Q \in \mathbb{F}_q^{m \times (m - 2\epsilon)}$, and is such that $Q^T$ generates an $[m, m - 2\epsilon]$ MDS code over $\mathbb{F}_q$. The candidate for the desired MRSC $\mathcal{C}_H$ is given by 
\begin{eqnarray}
\mathcal{C}_H & = & \left(\mathcal{C}_A^{(e)}\right)^{[n]}.
\end{eqnarray}
Note that $\mathcal{C}_H$ is an $[n, 2\epsilon]$ code. This completes the description of the construction. 
\end{constr}	

\vspace{0.1in} 

\begin{thm} \label{thm:p2p_example}
The $[n, 2\epsilon]$ code $\mathcal{C}_H$ obtained in Construction \ref{constr:mrc_Aspecial} is an MRSC of the $[n, m]$ code $\mathcal{C}_A$, where $A$ is given in \eqref{eq:Aspecial}.
\end{thm}
\begin{proof}
It is straightforward to check that the extended matrix $A^{(e)}$ in \eqref{eq:Aextended} satisfies the condition given in \eqref{eq:extend_property}, i.e., for any $S \subseteq [n], |S| = 2\epsilon$, we have 
\begin{eqnarray} 
\rho\left(A^{(e)}|_{S \cup \{n+1, \ldots, n + m - 2\epsilon\}}\right) & = & \rho\left(A|_{S}\right) + (m - 2\epsilon).
\end{eqnarray}
The proof now follows from the forward part of Lemma \ref{lem:MRSC_via_shortening}.
\end{proof}

\vspace{0.1in} 

\begin{note}
The converse part of Lemma \ref{lem:MRSC_via_shortening} can be used to show that given any $[n, 2\epsilon]$ MRSC $\mathcal{C}_H$ of $\mathcal{C}_A$, and if the vector $\bf{a}$ is such that  $a_i \neq 0 \  \forall i \in [K]$, then there exists an extended matrix $A^{(e)} = \left[ A \ | \ Q\right]$ such that $Q^T$ generates an $[m, m - 2\epsilon]$ MDS code over $\mathbb{F}_q$. Note that the assumption $a_i \neq 0$ was not imposed in Theorem \ref{thm:p2p_example}.
\end{note}

\vspace{0.1in} 

\begin{note} \label{rem:easy}
A generator matrix $H$ for the shortened code $\mathcal{C}_H$ as desired in Construction \ref{constr:mrc_Aspecial} can be obtained as follows. Let $W$ denote the $2\epsilon \times n$ generator matrix for the dual of the code generated by $Q^T$. Note that the $W$ generates an $[n, 2\epsilon]$ MDS code. The matrix $H$ is simply given by $H = WA^{(e)}$. The fact that $H$ generates the shortened code follows from the fact that $WQ = 0$. In fact, what this shows is that for the special case when the function $A$ is as given by \eqref{eq:Aspecial}, an optimal linear encoder can simply be constructed by pre-multiplying $A$ with the generator matrix of any $[n, 2\epsilon]$ MDS code. The reason why this works is as follows. If $E^n$ is any $\epsilon$ sparse vector, then $F^m = AE^n$ is also $\epsilon$-sparse, whenever $A$ takes on the form in \eqref{eq:Aspecial}. Thus, one can assume that the encoder has access to $Y^m + F^m$, the decoder has access to $Y^m = AX^n$, and is interested in computing $Y^m + F^m$. This is the special case of the function computation problem in Fig. \ref{fig:sys_model}, where the function $A$ is the identity matrix. It is well-known that for this special case, optimal encoding can be performed via the generator matrix of an MDS code. 

We note that above observation works with any scalar code (not necessarily Reed-Solomon code). However, the technique need not work while updating systems that use vector codes (like regenerating codes). In other words, Construction \ref{constr:mrc_Aspecial} (or the technique of pre-multiplying $A$ with an MDS matrix) is special to scalar codes. For vector codes, one needs to find out alternate ways of finding the extended code that satisfies Lemma \ref{lem:MRSC_via_shortening}  so that the shortened code generates the optimal encoder.  It is our hope that Lemma \ref{lem:MRSC_via_shortening} proves useful while constructing MRSCs while dealing with vector codes.
\end{note}

\section{Broadcasting to Receivers Interested in Different Functions} \label{sec:broadcast}
 
In the broadcast setting (see Fig. \ref{fig:sys_model_diff_function}), we consider two receivers that are interested in computing two separate linear functions of the message vector. The functions correspond to the matrices $A$ and $B$, respectively. Though the rest of the parameters are similar to those of the point-to-point case, we quickly describe them here again for sake of clarity.  The vector $X^n \in \mathbb{F}_q^n$ denotes the initial source message. The two receivers hold $AX^n$ and $BX^n$ as side information, respectively.  The matrices $A$ and $B$ have sizes $m_A \times n$ and $m_B \times n$, respectively, where $m_A, m_B \leq n$. Without loss of generality, we assume that $\rho(A) = m_A$ and $\rho(B) = m_B$. The updated source message  is given by vector $X^n + E^n$, where the difference-vector $E^n$ is $\epsilon$-sparse. Encoding is carried out via the $\ell \times n$ matrix $H$.
The goal is to recover the functions $A(X^n + E^n)$ and $B(X^n + E^n)$ at the respective receivers. 
The decoders used at the two receivers are denoted by $\mathcal{D}_1$ and  $\mathcal{D}_2$, respectively. Once again, we assume a zero-probability-of-error, worst-case-scenario model. We assume knowledge of the functions $A$, $B$ and the parameter $\epsilon$, while designing the encoder $H$. The communication cost for the model is given by $\ell$, assuming that the parameters $n, q, \epsilon, A$ and $B$ are fixed.  The triplet $(H, \mathcal{D}_1, \mathcal{D}_2)$ will be referred to as a {\em valid scheme} for the broadcast problem, if both the decoders' estimates are correct for all $X^n, E^n \in \mathbb{F}_q^n$ such that $\text{Hamming wt}.(E^n) \leq \epsilon$. 

Our goal in this section is to identify necessary and sufficient conditions on valid schemes for the broadcast problem. Specifically, we are interested in characterizing the minimum communication cost (among valid schemes) that can be achieved for the setting. We divide the discussion into three parts. We first consider the special case when the two linear codes $\mathcal{C}_A$ and $\mathcal{C}_B$ intersect trivially, i.e., $\mathcal{C}_A \cap \mathcal{C}_B = \{\bf{0}\}$. The optimal  communication cost for this case is straightforward to compute, given the observations from the point-to-point case. We then present necessary and sufficient conditions for the existence of sandwiched MRSCs. Recall that in the problem of sandwiched MRSCs, we begin with a code $\mathcal{C}_0$ and a subcode $\widehat{\mathcal{C}}$ of  $\mathcal{C}_0$. The goal is to identify a MRSC $\mathcal{C}$ of  $\mathcal{C}_0$ such that $\widehat{\mathcal{C}} \subset \mathcal{C}$.  Finally, we will show how the concept of sandwiched MRSCs can be used to identify valid schemes having optimal communication cost for the case of arbitrary matrices $A$ and $B$. Illustration of the results will be done by analyzing the settings in Examples  \ref{ex:bcast_nouse} and \ref{ex:bcast_use}.

\subsection{Optimal Communication Cost for  the Case $\mathcal{C}_A \cap \mathcal{C}_B = \{\bf{0}\}$}

Let $(H, \mathcal{D}_1, \mathcal{D}_2)$ denote a valid scheme for the case $\mathcal{C}_A \cap \mathcal{C}_B = \{\bf{0}\}$. For decoders $\mathcal{D}_1$ and $\mathcal{D}_2$ to be successful, we know from Theorem $\ref{thm:necessary}$ that $\text{dim}(\mathcal{C}_A \cap \mathcal{C}_H) \geq  \min(m_A, 2\epsilon)$ and $\text{dim}(\mathcal{C}_B \cap \mathcal{C}_H) \geq  \min(m_B, 2\epsilon)$, respectively. From this if follows that for the case when $\text{dim}(\mathcal{C}_A + \mathcal{C}_B) = \text{dim}(\mathcal{C}_A) + \text{dim}(\mathcal{C}_B)$, the communication cost $\ell = \text{dim}(\mathcal{C}_H)$ is lower bounded by 
\begin{eqnarray}
\ell & \geq & \min(m_A, 2\epsilon) + \min(m_B, 2\epsilon).
\end{eqnarray}
Given the achievability result from Section \ref{sec:achievability}, we see that the above bound is trivially achieved by encoding separately for the two receivers, where each of the two encodings is optimal for the respective receiver.  We formally state the above observations in the following theorem:

\vspace{0.1in}

\begin{thm} \label{thm:bcast_nouse}
The optimal communication cost for the broadcast setting in Fig. \ref{fig:sys_model_diff_function}  is given by 
\begin{eqnarray}
\ell & = & \min(m_A, 2\epsilon) + \min(m_B, 2\epsilon),
\end{eqnarray}
whenever the codes $\mathcal{C}_A $ and $\mathcal{C}_B$ intersect trivially. Achievability is guaranteed under the assumption that the field size $q > 2\epsilon n^{2\epsilon}$.
\end{thm}

\vspace{0.1in}

The following example illustrates the case under consideration. 

\vspace{0.1in}

\textit{Example \ref{ex:bcast_nouse} Revisited:}
Consider the setting in Example \ref{ex:bcast_nouse}, where  striping and encoding of a data file was done by an $[N, K]$ linear code over $\mathbb{F}_q$. Assume that the code is an $[N, K]$ MDS code with $K \geq 2$. We apply the broadcast model to simultaneously update the contents of two out of the $N$ storage nodes. Also, recall our assumption that the $K$-length coding vectors associated with the two storage nodes are given by ${\bf a} = [a_1 \ a_2 \ \ldots a_K], {\bf b} = [b_1 \ b_2 \ \ldots b_K], a_i, b_i \in \mathbb{F}_q, i \in [K]$, and the overall coding matrices $A$ and $B$ corresponding to the $m$ stripes are then given by 
\begin{eqnarray} \label{eq:ABspecial2}
A  = \left[ \begin{array}{cccc}
{\bf a} & & & \\ & {\bf a} & &  \\ & & \ddots & \\ & & & {\bf a}
\end{array} \right] &,  & 
B  =  \left[ \begin{array}{cccc}
{\bf b} & & & \\ & {\bf b} & &  \\ & & \ddots & \\ & & & {\bf b}
\end{array} \right],
\end{eqnarray} 
where both $A$ and $B$ have size $m \times mK$. Under the assumption that the $[N, K]$ code is MDS with $K \geq 2$, it is clear that $\bf{a} \neq \bf{b}$. In this case, it is straightforward to see that the codes $\mathcal{C}_A$ and $\mathcal{C}_B$ corresponding to the matrices $A$ and $B$ intersect trivially, i.e., $\mathcal{C}_A \cap \mathcal{C}_B = \{\bf{0}\}$. In this case, from Theorem \ref{thm:bcast_nouse}, we know that broadcasting does not help; the source must transmit as though it is individually updating the two storage nodes.

\subsection{Sandwiched Maximally Recoverable Subcodes}

We now take a slight detour, and identify necessary and sufficient conditions for the existence of sandwiched MRSCs. We assume that we are given an $[n, t]$ code $\mathcal{C}_0$ and an $[n, s]$ subcode $\widehat{\mathcal{C}}$ of $\mathcal{C}_0$. The question that we are interested is whether we can find an $[n, k]$ MRSC $\mathcal{C}$ of $\mathcal{C}_0$ such that $\widehat{\mathcal{C}}$ is a subcode of $\mathcal{C}$. We assume that the parameters $s, k, t$ satisfy the relation $ s \leq k \leq t$. Also, let  $G$, $G_0$ and $\widehat{G}$ denote generator matrices for the codes $\mathcal{C}$, $\mathcal{C}_0$ and $\widehat{\mathcal{C}}$, respectively. Recall from Definition \ref{defn:mrsc} that for $\mathcal{C}$ to be an MRSC of $\mathcal{C}_0$, it must be true that 
\begin{eqnarray}
\rho(G_0|_S) = k & \implies & \rho(G|_S) = k, \ S \subseteq [n], |S| = k.  \label{eq:mrc_def_rep}
\end{eqnarray}
Thus, a {\em necessary condition for the existence of the sandwiched MRSC} $\mathcal{C}$ can be given as follows:
\begin{eqnarray}
\rho(G_0|_S) = k & \implies & \rho\left(\widehat{G}|_S\right) = s, \ S \subseteq [n], |S| = k.  \label{eq:necessary_MRC_subcode}
\end{eqnarray}

Note that  the condition in \eqref{eq:necessary_MRC_subcode} is equivalent to saying that $\rho\left(\widehat{G}|_S\right) = s$ for any $S \subseteq [n], |S| = k$ which is an $k$-core of $\mathcal{C}_0^{\perp}$. In the following lemma, we show the sufficiency of the condition in \eqref{eq:necessary_MRC_subcode} for the existence of sandwiched MRSCs, under the assumption of a large enough finite field size $q$. 

\vspace{0.1in}

\begin{lem} \label{lem:MRSC_subcodes}
Suppose that we are given an $[n, t]$ code $\mathcal{C}_0$ over $\mathbb{F}_q$, and an $[n, s]$ subcode $\widehat{\mathcal{C}}$, whose generator matrices $G_0$ and $\widehat{G}$ satisfy \eqref{eq:necessary_MRC_subcode}, where $k$ is such that $s < k < t$. Then, there exists an $[n, k]$ code $\mathcal{C}$ such that 
\begin{itemize}
\item $\widehat{\mathcal{C}} \subseteq \mathcal{C}$ and 
\item $\mathcal{C}$ is a maximally recoverable subcode of $\mathcal{C}_0$,
\end{itemize}
whenever $q > {n \choose k}$.
\end{lem}
\begin{proof}
See Appendix \ref{app:proof_MRSC_subcodes}.
\end{proof}
 
\vspace{0.1in}
 
We next present an alternate, explicit construction of sandwiched MRSCs appears using linearized polynomials.

\vspace{0.1in}

\begin{constr} \label{constr:mrc_linearized_sandwiched}
Consider the $[n, t]$ code $\mathcal{C}_0$ over $\mathbb{F}_q$, having a generator matrix $G_0$. Also, consider  the $[n, s]$ subcode $\widehat{\mathcal{C}}$ of $\mathcal{C}_0$, having a generator matrix $\widehat{G}$. Without loss of generality, assume that the generator matrix $G_0$ is given by 
\begin{eqnarray} \label{eq:matrixB}
	G_0 & = & \left[ \begin{array}{c} \widehat{G} \\ B  \end{array} \right],
\end{eqnarray}
for some $(t - s) \times n$ matrix $B$. Note that we have $\rho(B) = t - s$. Let $\mathbb{F}_Q$ denote an extension field of $\mathbb{F}_q$, where $Q = q^{t-s}$.  Let $\mathcal{C}_0^{(Q)}$ and $\widehat{\mathcal{C}}^{(Q)}$ denote the $n$-length (linear) codes over  $\mathbb{F}_Q$ that are generated by $G_0$ and $\widehat{G}$, respectively. We note that $\text{dim}\left(\mathcal{C}_0^{(Q)} \right) = t$, and $\text{dim}\left(\mathcal{C}_0^{(Q)} \right) = s$. In this construction, we identify an $[n, k], s \leq k \leq t$ MRSC $\mathcal{C}$ of the $[n, t]$ code $\mathcal{C}_0^{(Q)}$ such that $\widehat{\mathcal{C}}^{(Q)} \subseteq \mathcal{C}$, whenever the matrices $G_0$ and $\widehat{G}$ satisfy \eqref{eq:necessary_MRC_subcode}. Toward this, let $\{\alpha_i \in \mathbb{F}_Q, 1 \leq i \leq t-s\}$ denote a basis of $\mathbb{F}_Q$ over $\mathbb{F}_q$. Define the elements $\{\beta_i \in \mathbb{F}_Q, 1 \leq i \leq n\}$ as follows:
\begin{eqnarray} \label{eq:beta_alpha_sw}
		[\beta_1 \ \beta_2 \ \cdots \ \beta_n] & = & [\alpha_1 \ \alpha_2 \ \cdots \ \alpha_{t-s}]B.
\end{eqnarray}
Next, consider the code $\mathcal{C}$ over $\mathbb{F}_Q$ having a generator matrix $G$ given by 
\begin{eqnarray} \label{eq:matrixG}
	G & = & \left[  \begin{array}{cccc} & & \widehat{G} &  \\ \hline \\ \beta_1 & \beta_2 & \cdots & \beta_n   \\
			\beta_1^q & \beta_2^q & \cdots & \beta_n^q \\
			& & \vdots &  \\
			\beta_1^{q^{k-s-1}} & \beta_2^{q^{k-s-1}} & \cdots & \beta_n^{q^{k-s-1}} \end{array} \right].
\end{eqnarray}
The code $\mathcal{C}$ is our candidate code, and this completes the description of the construction. In the following theorem, we prove that $\mathcal{C}$ is indeed an $[n, k]$ MRSC of $\mathcal{C}_0^{(Q)}$ such that $\widehat{\mathcal{C}}^{(Q)} \subseteq \mathcal{C}$.
\end{constr}	

\vspace{0.1in}

\begin{thm} \label{thm:mrc_linearized_sandwiched}
	Consider an $[n, t]$ code $\mathcal{C}_0$ over $\mathbb{F}_q$, having a generator matrix $G_0 \in \mathbb{F}_q^{t \times n}$. Also, consider  the $[n, s]$ subcode $\widehat{\mathcal{C}}$ of $\mathcal{C}_0$, having a generator matrix $\widehat{G}$ such that the following condition is satisfied:
	\begin{eqnarray}
		\rho(G_0|_S) = k & \implies & \rho\left(\widehat{G}|_S\right) = s, \ S \subseteq [n], |S| = k,  \label{eq:necessary_MRC_subcode_rep}
	\end{eqnarray}
	where $k$ is such that $s \leq k \leq t$. 
	Next, let $\mathcal{C}_0^{(Q)}$ and $\widehat{\mathcal{C}}^{(Q)}$ denote the $n$-length (linear) codes over  $\mathbb{F}_Q$ that are generated by $G_0$ and $\widehat{G}$, respectively, where $Q = q^{t-s}$. Then the code $\mathcal{C}$ over $\mathbb{F}_Q$ obtained in Construction \ref{constr:mrc_linearized_sandwiched} is an $[n, k]$ maximally recoverable subcode of $\mathcal{C}_0^{(Q)}$ such that $\widehat{\mathcal{C}}^{(Q)} \subseteq \mathcal{C}$.	
\end{thm}
\begin{proof}
	See Appendix \ref{app:mrc_linearized_sandwiched}.
\end{proof}

\subsection{Optimal Communication Cost for Arbitrary Matrices $A$ and $B$ under the Broadcast Setting}

We now use the concept of sandwiched MRSCs and characterize the optimal communication cost for the  broadcast problem when the codes $\mathcal{C}_A$ and $\mathcal{C}_B$ have a non-trivial intersection. Below, we first give a qualitative description of our approach, before presenting technical details. Consider a valid scheme $(H, \mathcal{D}_1, \mathcal{D}_2)$ that we design for the problem. Let $H_A$ and $H_B$ denote the rows in the row-space of $H$ which help the decoders $\mathcal{D}_1$ and $\mathcal{D}_2$ recover $A(X^n + E^n)$ and $B(X^n + E^n)$, respectively. Also let $\mathcal{C}_{H_A}$ and $\mathcal{C}_{H_B}$ denote the linear codes generated by $H_A$ and $H_B$ respectively. Our approach to minimizing the communication cost is to pick the encoder $H$ such that
\begin{itemize}
	\item $\mathcal{C}_{H_A}$ is a $2\epsilon$-dimensional (assuming that $m_A > 2\epsilon$) MRSC of  $\mathcal{C}_{A}$, 
	\item $\mathcal{C}_{H_B}$ is a $2\epsilon$-dimensional (assuming that $m_B > 2\epsilon$) MRSC of  $\mathcal{C}_{B}$, and
	\item also ``maximize" the dimension of intersection between the subcodes $\mathcal{C}_{H_A}$ and $\mathcal{C}_{H_B}$. The extent to which $\text{dim}(\mathcal{C}_{H_A} \cap \mathcal{C}_{H_B})$ can be maximized depends on the dimension of the intersection between $\mathcal{C}_{A}$ and $\mathcal{C}_{B}$.
\end{itemize}

For ease of presentation, we give the expression and proofs for optimal communication cost under the assumption\footnote{This is the hardest of all the cases. The remaining cases when one or both the ranks $m_A$ and $m_B$ are less than $2\epsilon$ can be similarly handled.} that both $m_A$ and $m_B$ are greater than $2\epsilon$. Under this assumption, consider the code $\widetilde{\mathcal{C}} = \mathcal{C}_A \cap \mathcal{C}_B$ having a generator matrix $\widetilde{H}$. Further, consider the quantities $\theta_A$, $\theta_B$ and $\theta$ defined as follows:

\begin{eqnarray}
\theta_A & = & \min_{\substack{S \subseteq [n], |S| = 2\epsilon \\ S \text{ is a } 2\epsilon\text{-core of }\mathcal{C}_A^{\perp}}}
\rho(\widetilde{H}|_S),  \label{eq:alphaA}\\
\theta_B & = & \min_{\substack{S \subseteq [n], |S| = 2\epsilon \\ S \text{ is a } 2\epsilon\text{-core of }\mathcal{C}_B^{\perp}}}
\rho(\widetilde{H}|_S),  \ \text{and} \label{eq:alphaB}\\
\theta & = & \min(\theta_A, \theta_B). \label{eq:alpha}
\end{eqnarray}
Note that the parameters $\theta_A, \theta_B$ and $\theta$ are entirely determined given the matrices $A$ and $B$, and the sparsity parameter $\epsilon$.  The following theorem characterizes the optimal communication cost in terms of the parameter $\theta$.

\vspace{0.1in}

\begin{thm}\label{thm:cost_broadcast}
The optimal communication cost for the broadcast setting shown in Fig. \ref{fig:sys_model_diff_function} is given by 
\begin{eqnarray}
\ell & = & 4\epsilon - \theta, 
\end{eqnarray}
whenever $m_A$ and $m_B$ are both greater than $2\epsilon$, and where the parameter $\theta$ is as defined by \eqref{eq:alpha}. Achievability is guaranteed under the assumption that the field size $q > \max\left( {n \choose 2\epsilon}, {n \choose \theta} \right)$.
\end{thm}
\begin{proof}
Let us prove the converse first, i.e., we show that the communication cost of an encoder $H$ associated with  any valid scheme is lower bounded by $\ell = \rho(H) \geq 4\epsilon - \theta$. Towards this, consider the codes $\mathcal{C}_H$, $\mathcal{C}_A$, $\mathcal{C}_B$, and define the codes $\mathcal{C}_{H_A}$, $\mathcal{C}_{H_B}$ and $\widehat{\mathcal{C}}$ as follows:
\begin{eqnarray}
\mathcal{C}_{H_A} & = & \mathcal{C}_H \cap \mathcal{C}_A, \\
\mathcal{C}_{H_B} & = & \mathcal{C}_H \cap \mathcal{C}_B, \ \text{and} \\
\widehat{\mathcal{C}} & = & \mathcal{C}_{H_A} \cap \mathcal{C}_{H_B}. 
\end{eqnarray}
Note that $\widehat{\mathcal{C}} \subseteq \widetilde{\mathcal{C}} =  \mathcal{C}_A \cap \mathcal{C}_B$. Also assume that $H_A$, $H_B$ and $\widehat{H}$ denote generator matrices for  the codes $\mathcal{C}_{H_A}$, $\mathcal{C}_{H_B}$ and $\widehat{\mathcal{C}}$, respectively. From Theorem \ref{thm:necessary}, we know that if $Y^n$ is any $2\epsilon$-sparse vector such that $AY^n \neq {\bf 0}$, then $H_AY^n \neq {\bf 0}$. From this it follows that  
\begin{eqnarray} \label{eq:temp_proof_1}
\rho(A|_S) = 2\epsilon \ \implies \rho(H_A|_S) = 2\epsilon, \ S \subseteq [n], |S| = 2\epsilon. 
\end{eqnarray}
Next, consider the definition of $\theta_A$ in \eqref{eq:alphaA}, and let $S^* \subseteq [n], |S^*| = 2\epsilon$ denote a $2\epsilon$-core of $\mathcal{C}_A^{\perp}$ such that 
\begin{eqnarray} \label{eq:temp_proof_2}
\rho(\widetilde{H}|_{S^*}) & = & \theta_A.
\end{eqnarray}
Now, consider the following chain of inequalities: 
\begin{eqnarray}
2\epsilon & \stackrel{(a)}{=} & \rho(A|_{S^*}) \\
& \stackrel{(b)}{=} & \rho(H_A|_{S^*}) \\
& \stackrel{(c)}{\leq} & \rho(\widehat{H}|_{S^*}) + \left(\rho(H_A) - \rho(\widehat{H}) \right) \\
& \stackrel{(d)}{\leq} & \rho(\widetilde{H}|_{S^*}) + \left(\rho(H_A) - \rho(\widehat{H}) \right) \\
& \stackrel{(e)}{=} & \theta_A + \left(\rho(H_A) - \rho(\widehat{H}) \right).
\end{eqnarray}
Here, $(a)$ follows because $S^*$ is a $2\epsilon$-core of $\mathcal{C}_A^{\perp}$, \newline
$(b)$ follows from \eqref{eq:temp_proof_1}, \newline
$(c)$ follows by observing  since $\widehat{\mathcal{C}}$ is a subcode of $\mathcal{C}_A$, we have $\rho(H_A) - \rho(\widehat{H}) \geq \rho(H_A|_{S}) - \rho(\widehat{H}|_{S})$, for any set $S \subset [n]$.  In particular, it is true that $\rho(H_A) - \rho(\widehat{H}) \geq \rho(H_A|_{S^*}) - \rho(\widehat{H}|_{S^*})$,  \newline
$(d)$ follows since $\widehat{\mathcal{C}}$ is a subcode of $\widetilde{\mathcal{C}}$, and finally \newline
$(e)$ follows from \eqref{eq:temp_proof_2}. \newline
In other words, we get that 
\begin{eqnarray}
\rho(\widehat{H}) & \leq & \theta_A + \rho(H_A) - 2\epsilon. \label{eq:rankH0_A}
\end{eqnarray}
In a similar fashion, one can also show that 
\begin{eqnarray}
\rho(\widehat{H}) & \leq & \theta_B + \rho(H_B) - 2\epsilon. \label{eq:rankH0_B}
\end{eqnarray}
The communication cost $\ell$ associated with the encoder $H$ can now be lower bounded as follows:
\begin{eqnarray}
\ell & = & \rho(H)\\
& \geq & \rho(H_A) + \rho(H_B) - \rho(\widehat{H}) \\
& \stackrel{(a)}{\geq} & \rho(H_A) + \rho(H_B) - \nonumber \\
&& \min(\theta_A + \rho(H_A) - 2\epsilon, \theta_B + \rho(H_B) - 2\epsilon ) \\
& \geq & \rho(H_A) + \rho(H_B) - \min(\theta_A , \theta_B)  - \nonumber \\
&& \max(\rho(H_A), \rho(H_B)) + 2\epsilon  \\
& \stackrel{(b)}{\geq} & 4\epsilon - \min(\theta_A , \theta_B)  = 4\epsilon - \theta.
\end{eqnarray}
Here $(a)$ follows from \eqref{eq:rankH0_A} and \eqref{eq:rankH0_B}, and $(b)$ follows from our assumption that the ranks of both $H_A$ and $H_B$ are greater than or equal to $2\epsilon$. This completes the proof of the lower bound on the communication cost.

\textit{Proof of Achievability:} Let us now show that it is indeed possible to construct a valid scheme having communication cost $\ell = 4\epsilon - \theta$, under the assumption of a sufficiently large field size. Towards this, consider the code $\widetilde{\mathcal{C}}$ and let $\widehat{\mathcal{C}}$ denote an $\theta$-dimensional MRSC of $\widetilde{\mathcal{C}}$. We know from Lemma \ref{lem:MRSC_subcodes} that the code $\widehat{\mathcal{C}}$ always exists\footnote{Note that this follows from Lemma \ref{lem:mrc_exist} as well; our use of Lemma \ref{lem:MRSC_subcodes} (with $\widehat{\mathcal{C}} = \{ {\bf 0} \}$ in Lemma \ref{lem:MRSC_subcodes}) must be considered as a matter of choice.} whenever the field size $q > {n \choose \theta}$. Also, let $\widehat{H}$ denote a generator matrix for the code $\widehat{\mathcal{C}}$. Now, observe that if $S$ is a $2\epsilon$-core of either $\mathcal{C}_A^{\perp}$ or $\mathcal{C}_B^{\perp}$, we know from the definition of $\theta$ in \eqref{eq:alpha} that $\rho(\widetilde{H}|_S) \geq \theta$. Noting that $\theta \leq 2\epsilon$ and using the fact that $\widehat{\mathcal{C}}$ is an $\theta$-dimensional MRSC of $\widetilde{\mathcal{C}}$, we get that $\rho(\widehat{H}|_S) = \theta$. In this case, we know from Lemma \ref{lem:MRSC_subcodes} that it is possible to identify a $2\epsilon$-dimensional MRSC $\mathcal{C}_{H_A}$ of $\mathcal{C}_{A}$ such that $\widehat{\mathcal{C}} \subseteq \mathcal{C}_{H_A}$, whenever the field size $q > {n \choose 2\epsilon}$. Similarly, we can also identify a $2\epsilon$-dimensional MRSC $\mathcal{C}_{H_B}$ of $\mathcal{C}_{B}$ such that $\widehat{\mathcal{C}} \subseteq \mathcal{C}_{H_B}$, whenever the field size $q > {n \choose 2\epsilon}$. The overall field size requirement, when we take into the account the minimum $q$ needed for the existence of $\widehat{\mathcal{C}}$ is given by $q > \max\left( {n \choose 2\epsilon}, {n \choose \theta} \right)$. The candidate code for the encoder is now given by $\mathcal{C}_H = \mathcal{C}_{H_A} + \mathcal{C}_{H_B}$.   
The communication cost of the encoder $H$ is then given by $\ell = \text{rank}(H_A) + \text{rank}(H_B) - \text{dim}(\mathcal{C}_{H_A} \cap \mathcal{C}_{H_B}) \leq \text{rank}(H_A) + \text{rank}(H_B) - \text{rank}(\widehat{H}) = 4\epsilon - \theta$. Also, we know from the achievability result of the point-to-point setting in Section \ref{sec:achievability} that decoders $\mathcal{D}_1$ and $\mathcal{D}_2$ in Fig. \ref{fig:sys_model_diff_function} can be constructed based on the matrices $H_A$ and $H_B$ respectively. This completes the proof of the achievability part of the theorem.
\end{proof}

\vspace{0.1in}

\begin{note}
In the above proof of achievability, we noted that $\widehat{\mathcal{C}} \subseteq \mathcal{C}_{H_A} \cap \mathcal{C}_{H_B}$.  In fact, it is straightforward to see that $\widehat{\mathcal{C}} = \mathcal{C}_{H_A} \cap \mathcal{C}_{H_B}$; else we would contradict the minimality of either $\theta_A$ or $\theta_B$ or both. 
\end{note}

\vspace{0.1in}

The following example illustrates the case under consideration.

\vspace{0.1in}

\textit{Example \ref{ex:bcast_use} revisited:} We now revisit Example \ref{ex:bcast_use} where we  considered striping and encoding of a data file by a $[N = 5, K = 3, D = 4] (\alpha = 4, \beta = 1)$ MBR code over $\mathbb{F}_q$ that encodes $9$ symbols into $20$ symbols and stores across $N = 5$ nodes such that each node holds $\alpha = 4$ symbols.
Recall that the contents of any $K = 3$ nodes are sufficient to reconstruct the $9$ uncoded symbols, and that the contents of any one of the $N = 5$ node can be recovered by connecting to any set of $D = 4$ other nodes. Also, recall that the encoding is described by a product-matrx construction as follows:
\begin{eqnarray}
C_{N \times \alpha} & = & \Psi_{N \times D} M_{D \times \alpha} \\
& = & \left[\begin{array}{c} \underline{\psi}_1 \\  \underline{\psi}_2 \\ \underline{\psi}_3 \\ \underline{\psi}_4 \\ \underline{\psi}_5 \end{array} \right] \left[ \begin{array}{cccc} m_1 & m_2 & m_3 & m_7 \\ m_2 & m_4 & m_5 & m_8 \\ m_3 & m_5 & m_6 & m_9 \\ m_7 & m_8 & m_9 & 0 \end{array} \right], \label{eq:MBR_prod_mtx_1}
\end{eqnarray}
where the encoding matrix $\Psi$ can be chosen as a Vandermonde matrix given as follows:
\begin{eqnarray}
\Psi & = & \left[ \begin{array}{cccc} 1 & \gamma & \gamma^2 & \gamma^3 \\  1 & \gamma^2 & \gamma^4 & \gamma^6 \\ 1 & \gamma^3 & \gamma^6 & \gamma^9 \\ 1 & \gamma^4 & \gamma^8 & \gamma^{12} \\ 1 & \gamma^5 & \gamma^{10} & \gamma^{15} \end{array} \right],
\end{eqnarray}
Here we pick $\gamma$ as a primitive element in $\mathbb{F}_q$. The $N \times \alpha$ matrix $C$ is the codeword matrix, with the $i^{\text{th}}$ row representing the contents that get stored in the $i^{\text{th}}$ node, $1 \leq i \leq 5$. It is known that any two nodes of an $[N = 5, K = 3, D = 4] (\alpha = 4, \beta = 1)$ MBR code share exactly one non-trivial linear combination of the respective symbols. For example, in the case of the product-matrix MBR code considered here, the first two nodes share the linear combination
\begin{eqnarray} 
\underline{\psi}_1 M \underline{\psi}_2^t & = & \underline{\psi}_2 M \underline{\psi}_1^t. \label{eq:prodmatrix_common}
\end{eqnarray}
The left-hand and the right-hand sides of \eqref{eq:prodmatrix_common} are linear combinations of the contents of the first and second nodes respectively, and the equality follows from the fact that the matrix $M$ is symmetric. 

Let us further recall the structure of the $4m \times 9m$ coding matrices $A$ and $B$ for the first two nodes, which are given by 
{\small 
\begin{eqnarray}
A = \left[ \begin{array}{cccc} A_1 & & & \\ & A_1 & & \\ & & \ddots & \\ & & & A_1 \end{array} \right] & , & B \ = \ \left[ \begin{array}{cccc} B_1 & & & \\ & B_1 & & \\ & & \ddots & \\ & & & B_1 \end{array} \right], \nonumber
\end{eqnarray}
}
where 

\begin{eqnarray}
A_1 \ = \ \left[ \begin{array}{ccccccccc} 1 & \gamma & \gamma^2 & & & & \gamma^3 & &  \\ 
& 1 & & \gamma & \gamma^2 & & &\gamma^3 &   \\ 
& & 1& & \gamma & \gamma^2 & & & \gamma^3 \\ 
& & & & & & 1 & \gamma & \gamma^2 \end{array}  \right] 
\end{eqnarray}
and 
\begin{eqnarray}
B_1 \ = \ \left[ \begin{array}{ccccccccc} 1 & \gamma^2 & \gamma^4 & & & & \gamma^6 & &  \\ 
& 1 & & \gamma^2 & \gamma^4 & & &\gamma^6 &   \\ 
& & 1& & \gamma^2 & \gamma^4 & & & \gamma^6 \\ 
& & & & & & 1 & \gamma^2 & \gamma^4 \end{array}  \right].
\end{eqnarray}

{\em Communication cost for updating two storage nodes via broadcasting :}
Let $\mathcal{C}_{A_1}$ and $\mathcal{C}_{B_1}$ denote the codes generated by the rows of $A_1$ and $B_1$. From \eqref{eq:prodmatrix_common}, we know that $\text{dim}\left(\mathcal{C}_{A_1} \cap  \mathcal{C}_{B_1}\right) = 1$. It is straightforward to see that $\mathcal{C}_{A_1} \cap  \mathcal{C}_{B_1}$ is generated by the vector 
\begin{eqnarray}
{\bf c } & = & [ 1, \ \gamma + \gamma^2, \ \gamma^2 + \gamma^4, \ \gamma^3, \ \gamma^4 + \gamma^5, \ \gamma^6, \nonumber \\
& & \ \gamma^3 + \gamma^6, \ \gamma^5 + \gamma^7, \ \gamma^7 + \gamma^8].
\end{eqnarray}
Note that if we satisfy the conditions 
\begin{eqnarray}
1 + \gamma & \neq & 0 \nonumber \\
1 + \gamma^2 & \neq & 0  \nonumber \\
1 + \gamma^3 & \neq & 0 \label{eq:notzero}
\end{eqnarray}
then, all the entires of the vector ${\bf c}$ are non-zeros. The conditions in \eqref{eq:notzero} are trivially satisfied for the field sizes that are needed for the broadcast model (e.g., even $q \geq 8$ is sufficient, since $\gamma$ is primitive.). Next, consider the codes $\mathcal{C}_{A}$ and $\mathcal{C}_{B}$  generated by the rows of $A$ and $B$, and observe that a generator matrix $\widehat{G}$ for the intersection 
$\mathcal{C}_{A} \cap \mathcal{C}_{B}$ is given by
\begin{eqnarray}
\widehat{G}_{m \times 9m} & = &  \left[ \begin{array}{cccc} {\bf c} & & & \\ & {\bf c} & & \\ & & \ddots & \\ & & & {\bf c} \end{array} \right].
\end{eqnarray}
We are now ready to calculate the quantities $\theta_A, \theta_B$ and $\theta$ in \eqref{eq:alphaA}-\eqref{eq:alpha}, and further apply Theorem \ref{thm:cost_broadcast} to calculate the communication cost for broadcasting. We assume that $m > 2\epsilon$. It is straightforward to see that the quantities $\theta_A, \theta_B$ and $\theta$ are given as follows:
\begin{eqnarray}
\theta_A & = & \min_{\substack{S \subseteq [n], |S| = 2\epsilon \\ S \text{ is a } 2\epsilon\text{-core of }\mathcal{C}_A^{\perp}}} 
\rho(\widetilde{H}|_S) \ = \ \left\lceil \frac{2\epsilon}{4}\right\rceil, \\
\theta_B & = & \min_{\substack{S \subseteq [n], |S| = 2\epsilon \\ S \text{ is a } 2\epsilon\text{-core of }\mathcal{C}_B^{\perp}}}
\rho(\widetilde{H}|_S) \ = \ \left\lceil \frac{2\epsilon}{4}\right\rceil,  \\
\theta & = & \min(\theta_A, \theta_B) \ = \ \left\lceil \frac{2\epsilon}{4}\right\rceil. 
\end{eqnarray}
The communication cost for broadcasting is then given by $\ell = 4\epsilon - \theta \approx 3.5\epsilon$. Note that if we individually update the two nodes, the overall communication will be $4\epsilon$. This completes our example.

\subsection{Extension to More than Two Destinations} \label{sec:3dest}

In this section, we briefly discuss possibility of extending the theory developed for $2$ destinations above, to more than $2$ destinations. In general when one employs an $[N, K]$ code such that any $K$ nodes can be used to recover the original data, one is interested in a theory that can extend to up to $K$ destinations. For updating more than $K$ destinations simultaneously, it is sufficient to use the encoding for updating the first $K$ destinations, since any $K$ nodes have all data.

It is straightforward to prove an extension of Theorem \ref{thm:bcast_nouse}, corresponding to the case when the $K$ codes, say $\mathcal{C}_{A_1}, \ldots, \mathcal{C}_{A_K}$, corresponding to the $K$ nodes do not have any pairwise intersection. In this case the overall communication cost is simply given by $\ell = \sum_{i = 1}^{K} \min(\rho(A_i), 2\epsilon)$. The case when there is non-trivial pair-wise intersection, but no triple intersection (i.e., $\mathcal{C}_{A_1} \cap \mathcal{C}_{A_2} \cap \mathcal{C}_{A_3} = \{ {\bf 0} \}$, etc.) can be analyzed along the lines of analysis done above for $2$ destinations in Theorem \ref{thm:cost_broadcast}. A lower bound on the communication cost quickly follows by using the lower bound in Theorem \ref{thm:cost_broadcast}, and is given by $\ell \geq 2K\epsilon - \sum_{i, j \in [K], i \neq j}\theta_{i, j}$, where $\theta_{i, j} = \min(\theta_{A_i}, \theta_{A_j})$, and where $\theta_{A_i}, 1 \leq i \leq K$ is defined along the lines of \eqref{eq:alphaA}. The achievability part, to show the optimality of the above lower bound seems non-trivial and requires work along two directions. The first direction involves a minor generalization of the notion of sandwiched MRSC, and this is rather straightforward. Recall that in Lemma \ref{lem:MRSC_subcodes}, we showed the existence of MRSCs containing a given subcode that satisfies a certain necessary condition, given by \eqref{eq:necessary_MRC_subcode}. We now need to show the existence of MRSCs that contain multiple non-intersecting subcodes; to be precise $K-1$ non-intersecting subcodes. One can take the direct sum of these $K-1$ subcodes, and it is sufficient to find an MRSC that contains the sum subcode. This can be done once again via Lemma \ref{lem:MRSC_subcodes}. The sum subcode must satisfy the necessary condition given by \eqref{eq:necessary_MRC_subcode}. The second direction is the hard part, where one needs to find the overall encoder along the lines of the proof of achievability in Theorem \ref{thm:cost_broadcast}. Recall that in the achievability proof, we first find a $\theta$-dimensional MRSC of the intersection, and then use the sandwiched MRSC result to arrive at the overall encoder. For the case of $K$ destinations, one possible approach is to first find individually $\theta_{i, j}$-dimensional MRSC of the respective intersection. We then find sum subcode of these individual MRSCs. The tricky part is to ensure that the sum subcode satisfies the necessary condition given by \eqref{eq:necessary_MRC_subcode} so that one can then apply Lemma \ref{lem:MRSC_subcodes} to get the overall encoder. While  $\theta_{i, j}$-dimensional MRSCs individually satisfy the necessary condition, it is not at all obvious as to how to ensure the same for the sum code. We leave this as a problem for future exploration; it is hoped that the above discussion will prove useful for a future research along this direction.

It is worth noting that the setting of Example \ref{ex:bcast_use} is an instance of the case described above, where there is pair-wise intersection, but no triple intersection. In fact, one can replace the specific MBR code in Example \ref{ex:bcast_use} with any $[N, K, D] (\alpha, \beta)$ MBR code with $K \geq 3$, and the resultant setting will still be instance of the same case with no triple intersection. This is because in any MBR code, any pair-wise intersection is of dimension $\beta$, and any three nodes must store $\alpha + (\alpha - \beta)  + (\alpha - 2\beta) = 3(\alpha - \beta)$ independent symbols. The last statement follows from the rank accumulation profile of MBR codes~\cite{non_achie_interior}. It is straightforward to check that any triple intersection would contradict the above observation. 

Finally, we chose not to address the case of multiple destination with non-trivial triple intersection, partly because of the difficulties mentioned above even without triple intersection, and also partly because of lack of knowledge of practical storage codes (with $K \geq 3$) that contain triple intersection.

\section{Conclusions} \label{sec:conc}

To conclude, in this paper, we considered a zero-error function update problem motivated by practical distributed storage systems, in which coded elements need to be updated by users whose data is out of sync with the data stored in the system. We studied the problem in its generality, by considering any possible linear update function. Optimal schemes for minimizing communication cost were obtained for point-to-point as well as broadcast settings. Examples based on well-known distributed storage codes were used to illustrate the potential applicability of the theory that we developed for both the point-to-point and broadcast cases. The most interesting contribution of this paper is perhaps the connection of the optimal solutions for the function update problem to the notion of maximally recoverable codes, which were conventionally studied as codes suitable for data storage applications. The current work provides an alternate view on this important class of codes, for which low field size constructions are still perused actively by many researchers. We introduced the class of sandwiched maximally recoverable codes, and showed its importance in constructing optimal schemes for the broadcast setting. 

The work opens up several interesting practically motivated questions: $1)$ Given the fact the low field size constructions of maximally recoverable codes are still largely unknown, can we obtain schemes for the function update problem which have sub-optimal communication costs, but have improved field sizes? In this regard, it will be interesting to explore connections with approximate variations of MRCs like sector-disk codes~\cite{blaum_mds}, \cite{pmds_twosectors_blaum}, \cite{shum_pmds} or partial maximally recoverable codes~\cite{balaji_pvk}. $2)$ Can we enhance the system model to one which does not have exact knowledge of the the sparsity parameter $\epsilon$? Recall that $\epsilon$ indicates the amount of update the source message undergoes. For instance, the source could encode assuming an approximate value of $\epsilon$, and also send out a low-cost hash function which could be used to detect failure of decoding at the receiver. Assuming a feedback link back to the source, we can perform further rounds of communication that incrementally update the destination function until decoding succeeds. $3)$ Can we use the results in this work to build secure incremental digital signature schemes, perhaps based on the McEliece crypto system~\cite{courtois2001achieve}? We note that in Fig. \ref{fig:sys_model}, if we assume that the source message vector and update vector are arbitrary, the encoder's output does not seem to reveal much about either of these vectors.

\bibliographystyle{IEEEtran}
\bibliography{sparse_updates}

\appendices

\section{Proof of Lemma \ref{lem:MRC_alt}} \label{app:MRC_eq}

Note that Part $1$ is the definition of MRSC given in Definition \ref{defn:mrsc}. To see the equivalence of Parts $1$ and $2$, we simply note that any set $S \subset [n], |S| = k$ is a $k$-core of $\mathcal{C}_0^{\perp}$ if and only if $\rho(G_0|_S) = k$. Next, to see why Part $2$ implies Part $3$, let $S$ denote any $k$-core of $\mathcal{C}_0^{\perp}$, and consider the matrix $H|_{[n]\backslash S}$. If $\rho\left(H|_{[n]\backslash S}\right) < n-k$, then since $\rho(H) = n - k$, this means $\rho\left(H|_{S}\right) > 0$, which contradicts the assumption in Part $2$ that $S$ is a $k$-core of $\mathcal{C}^{\perp}$ as well. Hence we conclude that Part $2$ implies $\rho\left(H|_{[n]\backslash S}\right) = n-k$. The fact that Part $3$ implies Part $2$ also follows similarly. Next we will prove the equivalence of Parts $1$ and $4$. Clearly, Part $4$ implies Part $1$, since in Part $4$ we consider all sets of cardinality less than or equal to $k$. To see why Part $1$ implies Part $4$, consider any set $S \subset [n], |S| = k - x$, such that $\rho(G_0|_S) = k - x$, for some $x, 1 \leq x \leq k - 1$. Since $\rho(G_0) = t > k$, there exists a set $S' \subset [n], |S'| = k - x$ such that $|S \cap S'| = 0$ and $\rho(G_0|_{S \cup S'}) = k$. Now, from Part $1$, we know that $\rho(G|_{S \cup S'}) = k$, which then implies that $\rho(G|_S) = k - x$. This completes the proof of Lemma \ref{lem:MRC_alt}.

\section{Proof of Theorem \ref{thm:mrc_linearized}} \label{app:mrc_linearized}

The following fact will be used in the proof of Theorem \ref{thm:mrc_linearized}.

\vspace{0.1in}

\begin{fact}[Problem $3.33$, \cite{roth}] \label{fact:roth}
Consider the finite field $\mathbb{F}_Q = \mathbb{F}_{q^t}$ as a vector space over $\mathbb{F}_q$, and let  $\{\alpha_i \in \mathbb{F}_Q, 1 \leq i \leq k\}$ denote a set of $k, k \leq t$ linearly independent elements over $\mathbb{F}_q$. Then the $k \times k$ matrix 
\begin{eqnarray}
\left[  \begin{array}{cccc} \alpha_1 & \alpha_2 & \cdots & \alpha_k   \\
\alpha_1^q & \alpha_2^q & \cdots & \alpha_k^q \\
& & \vdots &  \\
\alpha_1^{q^{k-1}} & \alpha_2^{q^{k-1}} & \cdots & \alpha_k^{q^{k-1}} \end{array} \right].
\end{eqnarray}
is invertible over $\mathbb{F}_Q$.
\end{fact}

\vspace{0.1in}

The above statement is equivalent to saying that a degree $q^{k-1}$- linearized polynomial $f(x) \in \mathbb{F}_Q[x]$ given by 
\begin{eqnarray}
f(x) & = & \sum_{i=0}^{k-1} a_ix^{q^{i}},  \ a_i \in \mathbb{F}_{q^t}
\end{eqnarray}
can be uniquely identified given its evaluations at $k$ points $\{\alpha_1, \cdots, \alpha_k\}$ that are linearly independent over $\mathbb{F}_q$.

\vspace{0.1in}

{\em Proof of Theorem \ref{thm:mrc_linearized}: }
First of all, we note that the code $\mathcal{C}$ obtained in Construction \ref{constr:mrc_linearized} is indeed a subcode of $\mathcal{C}_0^{(Q)}$. This follows from the fact that every row of the generator matrix $G$ of $\mathcal{C}$ can be written as an $\mathbb{F}_Q$-linear combination of the rows of $G_0$ as follows:
\begin{eqnarray}
[\beta_1^{q^{i-1}} \  \beta_2^{q^{i-1}}  \ \cdots \ \beta_n^{q^{i-1}} ] & = & [\alpha_1^{q^{i-1}} \ \alpha_2^{q^{i-1}} \ \cdots \ \alpha_t^{q^{i-1}} ]G_0, \nonumber \\
& & 1 \leq i \leq k.
\end{eqnarray}
The above equation follows by combining \eqref{eq:revise_1} and the fact that $G_0$ is a matrix over $\mathbb{F}_q$.

Next, without loss of generality, suppose that $\rho\left(G_0|_{[k]}\right) = k$. It then follows that the  elements $\{\beta_1, \cdots, \beta_k\}$ given by
\begin{eqnarray}
[\beta_1 \ \cdots \ \beta_k] & = & [\alpha_1 \  \cdots \ \alpha_t]G_0|_{[k]}, 
\end{eqnarray}
are linearly independent over $\mathbb{F}_q$. Combining the above statement along with Fact \ref{fact:roth}, we  get that the matrix $G|_{[k]}$ given by 
\begin{eqnarray}
G|_{[k]} & = & \left[  \begin{array}{cccc} \beta_1 & \beta_2 & \cdots & \beta_k   \\
\beta_1^q & \beta_2^q & \cdots & \beta_k^q \\
& & \vdots &  \\
\beta_1^{q^{k-1}} & \beta_2^{q^{k-1}} & \cdots & \beta_k^{q^{k-1}} \end{array} \right].
\end{eqnarray}
is invertible. This proves that $\text{dim}(\mathcal{C}) = k$, and also proves that $\mathcal{C}$ is an $[n, k]$ MRSC (see Definition~\ref{defn:mrsc}) of $\mathcal{C}_0^{(Q)}$.

\section{Proof of Converse Part of Lemma \ref{lem:MRSC_via_shortening}} \label{app:converse_MRC_extension}

Consider the $[n, t]$ code $\mathcal{C}_0$, and let $\mathcal{C}$ denote an $[n, k]$ MRSC of $\mathcal{C}_0$, where $k = t - \Delta$ for some $\Delta < t$. We will show that there exists an $[n + \Delta, t]$ extension $\mathcal{C}_0^{(e)}$ of $\mathcal{C}_0$ such that for any $S \subseteq [n], |S| = t - \Delta$, we have 
\begin{eqnarray} \label{eq:extend_property_1}
\rho\left(G_0^{(e)}|_{S \cup \{n+1, \ldots, n+\Delta\}}\right) & = & \rho\left(G_0|_{S}\right) + \Delta.
\end{eqnarray}
Here, $G_0$ and $G_0^{(e)}$ denote generator matrices of the codes $\mathcal{C}_0$ and $\mathcal{C}_0^{(e)}$, respectively. Towards this, consider the $(n-k) \times n$ parity check matrix $H$ of $\mathcal{C}$ given by 
\begin{eqnarray}
H & = & \left[ \begin{array}{c} H_0  \\ H_e \end{array} \right],
\end{eqnarray}
where $H_0$ is a parity check matrix of $\mathcal{C}$, and $H_e$ is some $(t - k) \times n$ matrix over $\mathbb{F}_q$ such that $\rho(H_e) = t - k$. We define the candidate code $\mathcal{C}_0^{(e)}$ for the extension as one whose parity check matrix $H_0^{(e)}$ is given by 
\begin{eqnarray}
H_0^{(e)} & = & \left[ \begin{array}{cc} H_0 & 0_{n-t \times \Delta}  \\ H_e & I_{\Delta} \end{array} \right].
\end{eqnarray}
First of all, to see that $\mathcal{C}_0^{(e)}$ is indeed an extension of $\mathcal{C}_0$, consider the code $\mathcal{C}'$ generated by $G_0^{(e)}|_{[n]}$, and  observe that 
\begin{eqnarray} \label{eq:temp_appB_1}
\mathcal{C}'^{\perp} & = & \mathcal{C}_0^{\perp}.
\end{eqnarray}
The above statement follows from the fact that the dual of the punctured code $\mathcal{C}'$ is the shortened code of the dual code. Clearly, \eqref{eq:temp_appB_1} implies that $\mathcal{C}'$ = $\mathcal{C}$, and hence $\mathcal{C}_0^{(e)}$ is an extension of $\mathcal{C}_0$. Next, we will show that $\mathcal{C}_0^{(e)}$ satisfies \eqref{eq:extend_property_1}. Towards this, consider any set $\mathcal{S} \subset [n], |S| = k' \leq k$ such that $\rho\left(G_0|_{S}\right) = k'$. We will prove that 
\begin{eqnarray} \label{eq:temp_appB_2}
\rho\left(G_0|_{S \cup \{n+1, \ldots, n+\Delta\}}\right) = k' + \Delta.
\end{eqnarray}
Note that proving \eqref{eq:temp_appB_2} is sufficient to prove \eqref{eq:extend_property_1}. We will prove \eqref{eq:temp_appB_2} by the  method of contradiction by supposing that $\rho\left(G_0|_{S \cup \{n+1, \ldots, n+\Delta\}}\right) < k' + \Delta$. In this case, we know that there exists a non-zero vector $\bf{c}^{(e)} \in \left(\mathcal{C}_0^{(e)}\right)^{\perp}$ having a support $T$ such that 
\begin{itemize}
	\item $T \subset S \cup \{n + 1, \ldots, n + \Delta\}$, 
	\item $T \cap \{n + 1, \ldots, n + \Delta\}$ is not empty, and
	\item $T \cap S$ is also not empty (this follows since $H$ has full row rank). 
\end{itemize}
Now, if we consider the vector $\bf{c} = \bf{c}^{(e)}|_{S}$, this would mean that the non-zero vector $\bf{c} \in \mathcal{C}^{\perp}$ (but clearly $\bf{c} \notin \mathcal{C}_0^{\perp}$). This then means that $\rho\left(G|_S\right) < k'$. Using Part $4.$ of Lemma \ref{lem:MRC_alt}, we see that this contradicts  our assumption that $\mathcal{C}$ is an $[n, k]$ MRSC of $\mathcal{C}_0$. Hence we conclude that $\eqref{eq:temp_appB_2}$ is indeed true, and this completes the proof of the converse.

\section{Proof of Lemma \ref{lem:MRSC_subcodes}} \label{app:proof_MRSC_subcodes}

Suppose that we are given an $[n, t]$ code $\mathcal{C}_0$ over $\mathbb{F}_q$, and an $[n, s]$ subcode $\widehat{\mathcal{C}}$ of $\mathcal{C}_0$ such that the following condition is satisfied: 
\begin{eqnarray}
\rho(G_0|_S) = k & \implies & \rho\left(\widehat{G}|_S\right) = s, \ S \subseteq [n], |S| = k, \label{eq:necessary_MRC_subcode_app}
\end{eqnarray}
where $k$ is such that $s < k < t$. We need to show that there exists an $[n, k]$ code $\mathcal{C}$ such that 
\begin{itemize}
	\item[(P$1$)] $\widehat{\mathcal{C}} \subset \mathcal{C}$ and 
	\item[(P$2$)] $\mathcal{C}$ is an MRSC of $\mathcal{C}_0$,
\end{itemize}
whenever $q > ((t - k)s + 1){n-1 \choose k}$. We use the equivalent definition of MRSCs given in Part $3$. of Lemma \ref{lem:MRC_alt} for proving the above statement. Let $H_0$ denote a  parity check matrix of the code $\mathcal{C}_0$. A parity check matrix $H$ for the code $\mathcal{C}$ can then be given by $H = \left[ \begin{array}{c} H_0 \\ \hline \\ \Delta H \end{array}   \right]$, where $\Delta H$ is a $(t-k) \times n$ matrix over $\mathbb{F}_q$. We will show that it is possible to identify $\Delta H$ under sufficiently large $q$ such that both (P$1$) and (P$2$) are satisfied.
Towards this, assume that the entries of $\Delta H$ are filled with the indeterminates $x_{i,j}, 1 \leq i \leq t-k, 1 \leq j \leq n$. Also, without loss of generality, assume that the generator matrix of the code $\widehat{\mathcal{C}}$ is given by $\widehat{G}  =  \left[ I_{s} \ | \ P_{s \times (n-s)} \right]$. Now, note that in order to satisfy (P$1$), it is sufficient if the $\{x_{i,j}\}$s satisfy the following relations:
\begin{eqnarray} \label{eq:inderminates}
\left[ \begin{array}{c} x_{i, 1} \\ x_{i, 2} \\ \vdots \\ x_{i, s} \end{array} \right] & = & -P_{s \times (n-s)} \left[ \begin{array}{c} x_{i, s+1} \\ x_{i, s+2} \\ \vdots \\ x_{i, n} \end{array} \right], \nonumber \\
&&  \forall \  i \in \{1, \ \cdots \ , \ t-k\}.
\end{eqnarray}
The above equation simply says that each of the $(t-k)$ $n$-length vectors $[x_{i, 1},  \ \ldots, \ x_{i,n}]^T, 1  \leq i \leq t-k$ must be picked from within the code $\widehat{\mathcal{C}}^{\perp}$.
Thus, let us assume that the $\{x_{i,j}\}$s are related as given by \eqref{eq:inderminates}. Next, in order to satisfy (P$2$), we know from Part $3$ of Lemma \ref{lem:MRC_alt} that it is sufficient if the following property is satisfied:
\begin{enumerate}[(P3)]
\item  for any $S \subseteq [n], |S| = k$ which is a $k$-core of $\mathcal{C}_0^{\perp}$, we have $\text{rank}(H|_{[n]\backslash S}) = n-k$.
\end{enumerate} 
Let us now see why it is always possible to pick the $\{x_{i,j}\}$s simultaneously satisfying \eqref{eq:inderminates} and (P$3$). Towards this, let $S$ denote a $k$-core of $\mathcal{C}_0^{\perp}$, and consider the $(n-k) \times (n-k)$ matrix  $H|_{[n]\backslash S}$. For satisfying (P$3$), we need to pick the $\{x_{i,j}\}$s such that $H|_{[n]\backslash S}$ is invertible. Of course, this can be done only if the matrix $H|_{[n]\backslash S}$ is not trivially rank deficient even before we pick the $\{x_{i,j}\}$s. We claim that there is no such trivial rank deficiency either due to the rows of the matrix $H_0$ or due to the relations in \eqref{eq:inderminates}. To see why this is the case, observe the following two points:
\begin{itemize}
\item Since $S$ is a $k$-core of $\mathcal{C}_0^{\perp}$, we have $\rho\left(H_0|_{[n]\backslash S}\right) = n-t$.
\item Secondly, if we restrict our attention to $\Delta H|_{[n]\backslash S}$, then the condition in \eqref{eq:necessary_MRC_subcode_app} ensures that there is no linear dependency among the columns of $\Delta H|_{[n]\backslash S}$, even when the $\{x_{i,j}\}s$ are picked such that \eqref{eq:inderminates} is satisfied. To see why this is true, suppose that the relations in \eqref{eq:inderminates} imply a linear dependency among the columns of $\Delta H|_{[n]\backslash S}$.  This is possible only if there exists a non-zero vector $\widehat{\bf{c}} \in \widehat{\mathcal{C}}$ such that $\text{supp}(\widehat{\bf{c}}) \subseteq [n]\backslash S$, which then implies that $\rho\left(\widehat{G}|_S\right) < s$. However this contradicts our assumption in \eqref{eq:necessary_MRC_subcode_app} that $\rho\left(\widehat{G}|_S\right) = s$, and we thus conclude that the relations in \eqref{eq:inderminates} do not introduce any linear dependency among the columns of $\Delta H|_{[n]\backslash S}$.
\end{itemize}
For the rest of the proof, let us define the multivariate polynomial $f_S(\{x_{i,j}\})$ as $f_S(\{x_{i,j}\}) = \text{det}(H|_{[n]\backslash S})$, and also let $f(\{x_{i,j}\}) = \Pi_{S \in \mathcal{T}} f_S(\{x_{i,j}\})$, where $\mathcal{T}$ denotes the collection of all $k$-cores of $\mathcal{C}_0^{\perp}$. We substitute for the $\{x_{i, j}, 1 \leq i \leq t - k, 1 \leq j \leq s \}$s in the polynomial $f(\{x_{i,j}\})$ using \eqref{eq:inderminates}, and then note that the degree of any $x_{i, j}, 1 \leq i \leq t - k, s+1 \leq j \leq n$ in the resultant polynomial is upper bounded by ${n \choose k}$. To see why this is true, we make the following two observations:
\begin{itemize}
	\item The degree of any $x_{i, j}, 1 \leq i \leq t - k, s+1 \leq j \leq n$ in any of the individual polynomials $f_S(\{x_{i,j}\})$ is at most $1$. This is true even after we substitute for the $\{x_{i, j}, 1 \leq i \leq t - k, 1 \leq j \leq s \}$s in the polynomial $f_S(\{x_{i,j}\})$ using \eqref{eq:inderminates}. The reason why the relations in \eqref{eq:inderminates} do not increase the degree of any individual $x_{i, j}$ is because of the fact that these relations only relate the elements of the same row in the matrix $\Delta H$, whereas any monomial in the determinant polynomial $f_S(\{x_{i,j}\})$ consists of $n-k$ terms that are picked from $n-k$ distinct rows. 
	\item The cardinality of the set $\mathcal{T}$ is upper bounded by ${n \choose k}$.
\end{itemize} 
The proof now follows by an application of Combinatorial Nullstellensatz~\cite{Alo}) which  says that there exists a non-zero evaluation of $f(\{x_{i,j}\})$ in $\mathbb{F}_q$, if $q > {n \choose k}$.

\section{Proof of Theorem \ref{thm:mrc_linearized_sandwiched}} \label{app:mrc_linearized_sandwiched}
First of all, it is clear that $\widehat{\mathcal{C}}^{(Q)} \subseteq  \mathcal{C}$, where $\mathcal{C}$ is as obtained in Construction \ref{constr:mrc_linearized}. Secondly, to see why the code $\mathcal{C}$ is a subcode of $\mathcal{C}_0^{(Q)}$, we note that the last $(k-s)$ rows of the generator matrix $G$ of $\mathcal{C}$ can be written as an $\mathbb{F}_Q$-linear combination of the rows of the matrix $B$, as follows:
\begin{eqnarray}
[\beta_1^{q^{i-1}} \  \beta_2^{q^{i-1}}  \ \cdots \ \beta_n^{q^{i-1}} ] & = & [\alpha_1^{q^{i-1}} \ \alpha_2^{q^{i-1}} \ \cdots \ \alpha_t^{q^{i-1}} ]B, \nonumber \\
& & 1 \leq i \leq k-s.
\end{eqnarray}
The above equation follows by combining \eqref{eq:beta_alpha_sw} and the fact that $B$ is a matrix over $\mathbb{F}_q$. Consequently, every row of $G$ is an $\mathbb{F}_Q$-linear combination of the rows of the matrix $G_0$, and hence $\mathcal{C}$ is indeed a subcode of $\mathcal{C}_0^{(Q)}$.

Next, without loss of generality, suppose that $\rho\left(G_0|_{S}\right) = k$, where $S = \{1, \ldots, k\}$.
We know from \eqref{eq:necessary_MRC_subcode_rep} that $\rho\left(\widehat{G}|_{S}\right) = s$. Once again, without loss of generality, assume that the matrix $\widehat{G}|_{S}$ is given by
\begin{eqnarray} \label{eq:Ghat}
\widehat{G}|_{S} & = & \left[I_{s} \ | \ \widehat{P}  \right],
\end{eqnarray}
for some $s \times (k - s)$ matrix $\widehat{P}$ over $\mathbb{F}_q$. Consider the  matrix $B$ introduced in \eqref{eq:matrixB}, and assume that $B|_S = [ B_1 \ \ B_2]$, where $B_1$ and $B_2$ are matrices over $\mathbb{F}_q$ of sizes $(t-s) \times s$ and $ (t-s) \times (k-s)$, respectively. The matrix $G_0|_S$ can now be written in terms of these sub matrices as follows:
\begin{eqnarray} \label{eq:G0S}
G_0|_{S} & = & \left[ \begin{array}{c|c}  I_{s} & \widehat{P} \\ \hline \\ B_1 &  B_2 \end{array} \right],
\end{eqnarray}
which can be row-reduced and written in the following form:
\begin{eqnarray} \label{eq:G0S_rr}
\text{row-reduced}\left(G_0|_{S}\right) & = & \left[ \begin{array}{c|c}  I_{s} & \widehat{P} \\ \hline \\ 0_{(t-s) \times s} &  B_2 - B_1\widehat{P} \end{array} \right].
\end{eqnarray}
We now make the following observations given the row-reduced form of $G_0|_S$ in \eqref{eq:G0S_rr}.
\begin{itemize}
\item Since $\rho\left(G_0|_{S}\right) = k$, it must be true that 
\begin{eqnarray} \label{eq:B2minusB1}
\rho\left( B_2 - B_1\widehat{P} \right) = k - s.
\end{eqnarray}
\item  Using the fact that $[\alpha_1, \cdots, \alpha_{t-s}]B = [\beta_1, \cdots, \beta_n]$ from \eqref{eq:beta_alpha_sw}, we get that 
\begin{eqnarray}
[\alpha_1, \cdots, \alpha_{t-s}]\left( B_2 - B_1\widehat{P} \right) & = & [\beta_{s+1}, \cdots, \beta_k]  - \nonumber \\ 
& & [\beta_{1}, \cdots, \beta_s]\widehat{P} \nonumber \\
& \triangleq & [ \gamma_1, \cdots, \gamma_{k-s}], \label{eq:gamma_in_terms_of_alpha}
\end{eqnarray} 
where we define the parameter $\gamma_i$ as follows:
\begin{eqnarray}
\gamma_i & = & \beta_{s+i} - \sum_{j = 1}^s \beta_j \widehat{P}(j, i), \ \ 1 \leq i \leq k - s.
\end{eqnarray}
Here $\widehat{P}(j, i)$ denotes the $(j, i)^{\text{th}}$ entry in the matrix $\widehat{P}$. Since $\widehat{P}$ is a matrix over  $\mathbb{F}_q$, we further note that 
\begin{eqnarray} \label{eq:gamma_pow_q}
\gamma_i^q & = & \beta_{s+i}^q - \sum_{j = 1}^s \beta_j^q \widehat{P}(j, i), \ \ 1 \leq i \leq k - s.
\end{eqnarray}
\end{itemize}
Let us next consider the matrix $G|_S$ and show that $\rho\left( G|_S \right) = k$. Towards this, first of all note from \eqref{eq:matrixG} and \eqref{eq:Ghat} that $G|_S$ can be written as follows:
\begin{eqnarray} \label{eq:G|S}
G|_S &  = & \left[ \begin{array}{c|c}  I_{s} & \widehat{P} \\ \hline \\ \Lambda_1 &  \Lambda_2 \end{array} \right],
\end{eqnarray}
where 
\begin{eqnarray}
\Lambda_1 \ = \ \left[  \begin{array}{cccc} \beta_1 & \beta_2 & \cdots & \beta_{s}   \\
\beta_1^q & \beta_2^q & \cdots & \beta_{s}^q \\
& & \vdots &  \\
\beta_1^{q^{k-s-1}} & \beta_2^{q^{k-s-1}} & \cdots & \beta_{s}^{q^{k-s-1}} \end{array} \right] 
\end{eqnarray}
and 
\begin{eqnarray}
\Lambda_2 \ = \ \left[  \begin{array}{cccc} \beta_{s+1} & \beta_{s+2} & \cdots & \beta_k   \\
\beta_{s+1}^q & \beta_{s+2}^q & \cdots & \beta_k^q \\
& & \vdots &  \\
\beta_{s+1}^{q^{k-s-1}} & \beta_{s+2}^{q^{k-s-1}} & \cdots & \beta_k^{q^{k-s-1}} \end{array} \right].
\end{eqnarray}
The matrix $G|_S$ in \eqref{eq:G|S} can be row-reduced and written as follows:
\begin{eqnarray} \label{eq:G|Srr}
\text{row-reduced}\left(G|_S\right) &  = & \left[ \begin{array}{c|c}  I_{s} & \widehat{P} \\ \hline \\ 0_{(k-s) \times s} &  \Lambda_2 - \Lambda_1 \widehat{P} \end{array} \right].
\end{eqnarray}
In order to prove that $\rho\left( G|_S \right) = k$, note that it is enough if we show that $\rho\left( \Lambda_2 - \Lambda_1 \widehat{P} \right) = k-s$. Towards this, observe that $\Lambda_2 - \Lambda_1 \widehat{P}$ can be written as in \eqref{eq:lamdadiff}. Equation 
\begin{figure*}
\hrule
\begin{eqnarray} \label{eq:lamdadiff}
\Lambda_2 - \Lambda_1 \widehat{P} & = &  \left[ \begin{array}{c} \left[\beta_{s+1}, \cdots,  \beta_{k}\right] -  \left[\beta_{1}, \cdots,  \beta_{s}\right]\widehat{P}   \\
\left[\beta_{s+1}^q, \cdots,  \beta_{k}^q\right] -  \left[\beta_{1}^q, \cdots,  \beta_{s}^q\right]\widehat{P} \\
\vdots \\
\left[\beta_{s+1}^{q^{k-s-1}}, \cdots,  \beta_{k}^{q^{k-s-1}}\right] -  \left[\beta_{1}^{q^{k-s-1}}, \cdots,  \beta_{s}^{q^{k-s-1}}\right]\widehat{P}
\end{array}   \right] \nonumber \\
& = & \ \left[  \begin{array}{cccc} \gamma_1 & \gamma_2 & \cdots & \gamma_{k-s}   \\
\gamma_1^q & \gamma_2^q & \cdots & \gamma_{k-s}^q \\
& & \vdots &  \\
\gamma_1^{q^{k-s-1}} & \gamma_2^{q^{k-s-1}} & \cdots & \gamma_{k-s}^{q^{k-s-1}} \end{array} \right] \label{eq:temp_ss_proof_1},
\end{eqnarray}
\end{figure*}
where \eqref{eq:temp_ss_proof_1} follows from \eqref{eq:gamma_pow_q}.  Recall that the $\{\gamma\}$s are related to the  $\{\alpha\}s$ as given in \eqref{eq:gamma_in_terms_of_alpha}. From \eqref{eq:B2minusB1} we know that $\rho\left( B_2 - B_1\widehat{P} \right) = k - s$, from which it follows that the $\{\gamma\}$s are linearly independent over $\mathbb{F}_q$. Now, we can apply Fact \ref{fact:roth} to see that the matrix $\Lambda_2 - \Lambda_1 \widehat{P}$ in \eqref{eq:temp_ss_proof_1} in indeed invertible, i.e., $\rho\left( \Lambda_2 - \Lambda_1 \widehat{P} \right) = k-s$. From this it follows that $\rho\left(G|_S\right) = k$, and this completes our proof.

\begin{IEEEbiographynophoto}{N. Prakash} received his Bachelors (Electronics and Communication Engineering), Masters (Electrical Engineering) and PhD degrees (Electrical Communication Engineering) from College of Engineering, Thiruvananthapuram, Indian Institute of Technology (IIT) Madras, and Indian Institute of Science (IISc), Bangalore in 2004, 2006 and 2015, respectively. He received the 2006 Siemens prize for the best academic performance, Masters Electrical Engineering, IIT-Madras. His industry experience includes role as an algorithm design engineer, Beceem Communications (2004-2006), and multiple internships with the Advanced Technology Group of NetApp, Bangalore.  He is currently a postdoctoral associate at the Department of Electrical Engineering and Computer Science, Massachusetts Institute of Technology (MIT), Cambridge, USA. His current research is broadly focused on the design and implementation of algorithms for  networked-storage and distributed computing platforms. Specifically, he is interested in the design of algorithms for consistent object storage in asynchronous platforms, erasure coding for fault-tolerant storage, data compression and security. 
\end{IEEEbiographynophoto}

\begin{IEEEbiographynophoto}{Muriel M{\'{e}}dard} is the Cecil H. Green Professor in the Electrical Engineering and Computer Science (EECS) Department at MIT and leads the Network Coding and Reliable Communications Group at the Research Laboratory for Electronics at MIT. She has co-founded three companies to commercialize network coding, CodeOn, Steinwurf and Chocolate Cloud. She has served as editor for many publications of the Institute of Electrical and Electronics Engineers (IEEE), of which she was elected Fellow, and she has served as Editor in Chief of the IEEE Journal on Selected Areas in Communications. She was President of the IEEE Information Theory Society in 2012, and served on its board of governors for eleven years. She has served as technical program committee co-chair of many of the major conferences in information theory, communications and networking. She received the 2009 IEEE Communication Society and Information Theory Society Joint Paper Award, the 2009 William R. Bennett Prize in the Field of Communications Networking, the 2002 IEEE Leon K. Kirchmayer Prize Paper Award and several conference paper awards. She was co-winner of the MIT 2004 Harold E. Edgerton Faculty Achievement Award, received the 2013 EECS Graduate Student Association Mentor Award and served as Housemaster for seven years. In 2007 she was named a Gilbreth Lecturer by the U.S. National Academy of Engineering. She received the 2016 IEEE Vehicular Technology James Evans Avant Garde Award, the 2017 Aaron Wyner Distinguished Service Award from the IEEE Information Theory Society and the 2017 IEEE Communications Society Edwin Howard Armstrong Achievement Award.
\end{IEEEbiographynophoto}

\end{document}